\documentclass[letterpaper]{article} 
\usepackage{aaai25}  
\usepackage{times}  
\usepackage{helvet}  
\usepackage{courier}  
\usepackage[hyphens]{url}  
\usepackage{graphicx} 
\usepackage{natbib}  
\usepackage{caption} 
\usepackage{algorithm}
\usepackage{algorithmic}

\usepackage{newfloat}
\usepackage{listings}
\DeclareCaptionStyle{ruled}{labelfont=normalfont,labelsep=colon,strut=off} 
\floatstyle{ruled}
\newfloat{listing}{tb}{lst}{}
\floatname{listing}{Listing}
\usepackage{cancel}
\usepackage{amsmath}
\usepackage{amsthm}
\usepackage{xspace}
\usepackage{wasysym}
\usepackage{amssymb}

\newcommand{\PaperVersion}{arxiv} 
\newcommand{\OnlyAaai}[1]{\ifthenelse{\equal{\PaperVersion}{aaai}}{#1}{}}
\newcommand{\OnlyArxiv}[1]{\ifthenelse{\equal{\PaperVersion}{arxiv}}{#1}{}}
\OnlyArxiv{\nocopyright} 

\newcommand{\tagLabel}[2]{\tag{\textbf{#1}}\label{#2}}

\newcommand{\rewa}[1] {\mathsf{good}_{#1}  }
\newcommand{\punis}[1] {\mathsf{bad}_{#1}  }
\newcommand{\st}{S}
\newcommand{\desbaseset}[1]{D_{#1}^+ }
\newcommand{\undesbaseset}[1]{D_{#1}^- }
\newcommand{\atmset}{\ensuremath{\mathit{Atm}}\xspace}
\newcommand{\PROP}{\mathit{Atm}}
\newcommand*{\prop}{p}
\newcommand{\classbelbase}{\mathbf{M} }
\newcommand{\classbelbaseuniv}{\mathbf{M} }
\newcommand{\lang}{\ensuremath{\mathcal{L}}}
\newcommand{\langdyn}{\lang_{\mathsf{dyn} } } 
\newcommand{\langprog}{\ensuremath{\mathcal{L}}_{\mathsf{prog} } }
\newcommand{\union}{\cup}
\newcommand{\set}[1]{{\{#1\}}}
\newcommand{\suchthat}{:}

\newcommand{\bigand}{\bigwedge}
\newcommand{\lbigand}{\bigand}
\newcommand{\bigor}{\bigvee}
\newcommand{\limply}{\rightarrow}
\newcommand{\lequiv}{\leftrightarrow}

\newcommand{\agtset}{\ensuremath{\textit{Agt}}\xspace}
\newcommand{\AGT}{\mathit{Agt}}

\newcommand{\expbel}[1] {\triangle_{#1}   }
\newcommand{\impbel}[1] {\Box_{#1}  }

\newcommand{\bnf}{::=}

\newcommand{\fraglang}{ \mathcal{L}_{0} }
\newcommand{\stateval}{\mathit{V}}
\newcommand{\belbaseset}{\mathit{B}}
\newcommand{\relstate}[1]{\mathcal{E}_{#1}}
\newcommand{\relattr}[1]{\mathcal{A}_{#1}}
\newcommand{\relrepu}[1]{\mathcal{R}_{#1}}
\newcommand{\reldyn}[1]{\mathcal{P}_{#1}}
\newcommand{\iconstraint}{\mathit{U}}
\newcommand{\setbelbase}{\mathbf{S}}

\newcommand{\defin}{~\stackrel{\mbox{\scriptsize def}}{=}~} 

\newcommand{\sat}{\models }

\newcommand{\supervarphi}{\varphi_0}
\newcommand{\relepist}[1]{\mathcal{R}_{#1}}

\newcommand{\attract}[2]{\textnormal{\smiley}_{#1}^{#2}  }
\newcommand{\repuls}[2]{\textnormal{\frownie}_{#1}^{#2}  }
\newcommand{\attractreal}[2]{[\textnormal{\smiley}]_{#1}^{#2}  }
\newcommand{\repulsreal}[2]{[\textnormal{\frownie}]_{#1}^{#2}  }

\newcommand{\motiv}[1] {\mathsf{M}^{\uparrow}_{#1}  }
\newcommand{\demotiv}[1] {\mathsf{M}^{\downarrow}_{#1}  }
\newcommand{\realmotiv}[1] {\mathsf{RM}^{\uparrow}_{#1}  }
\newcommand{\realdemotiv}[1] {\mathsf{RM}^{\downarrow}_{#1}  }
\newcommand{\indiff}[1] {\mathsf{I}_{#1}  }
\newcommand{\rindiff}[1] {\mathsf{RI}_{#1}  }
\newcommand{\ambi}[1] {\mathsf{A}_{#1}  }
\newcommand{\rambi}[1] {\mathsf{RA}_{#1}  }
\newcommand{\jokertwo }  { \blacktriangleleft  }

\newcommand{\setbelbasecustom}{\setbelbase_{\setrelevantformulas}}
\newcommand{\setrelevantformulas}{\Gamma}

\newtheorem{theorem}{Theorem}
\newtheorem{proposition}{Proposition}
\newtheorem{definition}{Definition}
\newtheorem{lemma}{Lemma}
\newtheorem{example}{Example}

\title{A Computationally Grounded Framework for Cognitive Attitudes%
    \OnlyArxiv{\\(extended version)}}
\author {
    Tiago de Lima\textsuperscript{\rm 1},
    Emiliano Lorini\textsuperscript{\rm 2},
    Elise Perrotin\textsuperscript{\rm 3},
    François Schwarzentruber\textsuperscript{\rm 4}
}
\affiliations {
    \textsuperscript{\rm 1}CRIL, Univ Artois and CNRS, Lens, France\\
    \textsuperscript{\rm 2}IRIT, CNRS, Toulouse University, France\\
    \textsuperscript{\rm 4}National Institute of Advanced Industrial Science and Technology, Tokyo, Japan\\
    \textsuperscript{\rm 3}ENS Lyon, LIP, France\\
    delima@cril-lab.fr,
    Emiliano.Lorini@irit.fr,
    elise.perrotin@aist.go.jp,
    francois.schwarzentruber@ens-lyon.fr
}

\newcommand{\textif}{\text{ if }}

\begin{document}

\maketitle

\begin{abstract}
We introduce  a novel language for reasoning about agents' cognitive attitudes
 of both epistemic and motivational type. We interpret it by means of a  computationally grounded semantics using belief bases.
 Our language includes five types of modal operators for implicit belief,
 complete attraction, complete repulsion, realistic attraction
 and realistic repulsion.
 We give an axiomatization 
 and show that our operators are not mutually expressible and that they can be combined to represent  a large variety of psychological concepts including ambivalence, indifference, being motivated, being demotivated and preference. 
 We present a dynamic extension of the language that supports
 reasoning about the effects of belief change operations.   Finally,
we provide a succinct formulation of model checking for our languages 
and a PSPACE model checking algorithm 
relying  on a reduction into TQBF.
We present    some experimental
results for the implemented algorithm on computation time in a concrete example.
\end{abstract}

\begin{links}
    \link{Code}{https://gitlab.in2p3.fr/tiago.delima/cognitive-attitudes-source-code/}%
    \label{code-url}
    \OnlyAaai{\link{Extended version}{https://???}%
        \label{extended-url}}
\end{links}

\section{Introduction}

An agent's cognitive
state
encompasses her
epistemic attitudes
(e.g., beliefs)
and 
motivational (or conative)  attitudes
(e.g., desires and preferences).
Their relationships 
as well as their influence on 
the agent's 
behavior are objects of study in cognitive psychology 
\cite{AlbarracinHandbook2018} and philosophy of mind \cite{Sea01,HumberstoneFit}.
They play a prominent role
in the explanation of
others' behaviours and of our own behaviours
through the so-called
intentional stance \cite{DennettIntStance}. 
Cognitive (or mental) attitudes have also been studied by logicians, both in philosophy and in AI. Several logics dealing with epistemic and practical reasoning of rational agents have been proposed. 
These include epistemic logics \cite{Hintikka,Fagin1995},
logics of preferences \cite{WrightPreference,WrightPreference2,LiuBook,BenthemGirardRoy},
logics of beliefs and preferences \cite{Boutilier94,DBLP:journals/tplp/Lorini21,DBLP:journals/corr/abs-2101-00485}, 
logics of desires
and pro-attitudes \cite{DBLP:journals/mima/DuboisLP17,DBLP:conf/aaai/SuSLR07},
logics of intention \cite{Shoham2009,IcardPacuit,LoriniHerzigSynthese}, 
BDI (belief, desire, intention) logics \cite{Coh90,DBLP:conf/kr/HerzigL04,Mey99,Woo00}. 
The idea of describing rational agents in terms of their epistemic and motivational
attitudes is 
shared with classical decision theory and game
theory,
according to which
rational agents are assumed to make decisions on the basis
of their  beliefs and preferences. 

Most logics of cognitive attitudes rely on extensional semantics based on so-called Kripke models: sets of possible worlds (or states) supplemented with one or more binary
accessibility 
relations for each agent representing, e.g.,  the agent's epistemic uncertainty or preference ordering over the states.  Multi-relational
Kripke models are general and mathematically
elegant. Nonetheless, they are limited from
a modeling point of view as they are not succinct. Even a simple situation
like a   card game (e.g.,  Hanabi) with
4 players having 5 cards each among a set of 50 cards
require a Kripke model
with $5.5 \times 10^{23}$ states
to represent all alternatives  on which 
an agent's uncertainty and  preferences bear.
For this reason, 
it is hard, if not unfeasible,  to implement model checking
for a logic of cognitive attitudes using a Kripke semantics
since the model cannot be explicitly constructed. 
Some alternatives
to solve this succinctness problem have been proposed 
in the area of epistemic logic 
including 
semantics based on 
BDDs \cite{DBLP:journals/logcom/BenthemEGS18}, Boolean formulas and programs \cite{DBLP:conf/atal/CharrierS15,DBLP:conf/atal/CharrierS17,DBLP:conf/aiml/CharrierS18,DBLP:journals/logcom/CharrierPS19},
and the notion of visibility \cite{DBLP:journals/ai/CooperHMMPR21,DBLP:conf/atal/HoekTW11}.
However, to our knowledge no general
succinct semantics
for a  logic
of cognitive attitudes, combining epistemic
ones with motivational ones, has been proposed up to now. 

In this paper, 
we tackle this problem by relying on 
a semantics for cognitive attitudes based on the notion
of belief base that
has been shown to provide an interesting
computationally grounded alternative to Kripke models
in representing epistemic concepts
in the single-agent case 
\cite{KONOLIGE}
as well as in the multi-agent case 
\cite{LoriniAAAI2018,DBLP:journals/corr/abs-1907-09114,LoriniAI2020,LoriniRapionAAMAS2022}.
Unlike standard Kripke semantics
for  logics of cognitive attitudes
in which possible states and 
agents' 
accessibility relations
are given as primitive, 
in our  semantics
they 
are  defined    from and grounded on  the primitive
concept
of belief base.
We will use belief bases
to define 
three types of accessibility relation 
capturing, respectively,
the states that an agent
considers possible,
those she finds attractive and
those she finds repulsive. The first belongs to the epistemic
sphere, while the second and the third
belong to the motivational sphere. 
By means of this semantics, we will  interpret
a modal language
including five 
modal operators for implicit belief,
 complete attraction, complete repulsion, realistic attraction
 and realistic repulsion
 as well as a dynamic
 extension of it that 
supports
 reasoning about the effects
 of belief change operations. 
 We will show that 
 in our semantics
model checking can be formulated
in a more succinct way
than in the standard Kripke semantics,
which opens up the possibility
of using it in practice. 
We will  provide a PSPACE model checking procedure
relying  on a reduction into TQBF
and   some experimental
results for the implemented procedure on computation time.

\paragraph{Related work} Our belief base approach  has some   aspects
in common with the body of literature
on preference representation  based on priority graphs
\cite{JonghLiu2009,LiuBook,SouzaMoreira2021,SouzaThesis},
that was recently   extended to deontic logic 
in \cite{BenthemGrossiLiu2014}.
We share with these works
 the idea  i) that there are two alternative approaches to the representation of 
 preferences and more generally of mental attitudes, namely the extensional approach based, e.g., on multi-relational Kripke models and the syntactic approach based, e.g., on belief bases or priority graphs, and   ii) that  the two approaches must coexist and be integrated into the same logical framework.
 But there are also some important differences. 
 These works are  mostly focused on the single-agent case while our approach is multi-agent. 
 Secondly,  they mainly focus on preference representation, while our emphasis is on the bipolar aspect of motivational  attitudes
(i.e., attraction vs repulsion, motivation vs demotivation). 
Thirdly,
our modal language of cognitive attitudes is profoundly different from existing ones,  including from the one presented 
in \cite[Chapters 4 and 5]{SouzaThesis}
which, like ours, covers both epistemic and motivational attitudes. Souza's language 
 includes ``betterness'' (S4) 
and ``strict betterness'' modalities,
while our language includes primitive modalities for
(complete and realistic) attraction and repulsion.
Our language 
also includes explicit belief modalities
and the axiomatization we will provide has specific axioms relating explicit belief to implicit mental attitudes. Souza's language has no explicit belief involved.  

All proofs and additional content can be found in 
\OnlyAaai{the extended version of the paper (see URL on page~\pageref{extended-url}).}
\OnlyArxiv{the appendix.}

\section{Belief base semantics}%
\label{sec:semantics}

Following  \cite{LoriniAAAI2018,LoriniAI2020},
in this section
we present a formal semantics
for cognitive attitudes 
exploiting  belief bases,
where elements
of an agent's belief base
are explicit beliefs of the agent. 
Basic notions of appetitive
and aversive desire
are directly  defined from
a belief base.
We will use belief bases to define 
three types of accessibility
relations 
for belief, attraction and repulsion.

\subsection{States}

Assume
a countably  infinite set  of atomic propositions $\PROP$
and
 a finite set of agents $\AGT = \{ 1, \ldots, n \}$.
 We suppose the set $\PROP$
 includes special atomic formulas
 of type $\rewa{i}$
 and $\punis{i}$
 for every $i \in \AGT$
 meaning respectively that 
 ``agent $i$ gets a reward''
 (or
 ``agent $i$ feels pleasure'')
 and
``agent $i$ gets a punishment''
 (or ``agent $i$ feels pain'').
We define the language $\fraglang$ for
explicit belief 
by the following grammar in Backus-Naur Form (BNF):
\[ \fraglang\defin \alpha \bnf \prop
                          \mid \lnot\alpha
                          \mid \alpha \land \alpha
                          \mid \expbel{i}\alpha   , \]
where $p$ ranges over $\atmset$
and $i$ ranges over $\agtset$.
$\fraglang$ is the language used to represent explicit beliefs.
The formula $\expbel{i} \alpha$ reads ``agent $i$ has the explicit belief that $\alpha$''. 

In our semantics, a state is not a primitive notion but is decomposed into different elements:
one belief base per agent
and a valuation of propositional atoms. 
An agent's belief base is  analogous to the agent's local state
and
the propositional valuation is analogous to the local
state of the environment
in interpreted systems \cite{Fagin1995}. 
\begin{definition}[State]%
\label{state}
A state is a
tuple $\st = \big( (\belbaseset_i)_{i \in \agtset},\stateval \big)$
where $\belbaseset_i \subseteq \fraglang$
is agent $i$'s belief base, 
and 
$\stateval \subseteq \atmset$ is the actual environment.
The set of all states is noted $\setbelbase$.
\end{definition}

The following definition specifies truth conditions for formulas in $\fraglang$ relative to states. 
\begin{definition}[Satisfaction relation]%
\label{satrel}
Let 
$\st = \big( (\belbaseset_i)_{i \in \AGT}, \stateval \big) \in \setbelbase$.
We define $\st \models \alpha$ by:
\begin{eqnarray*}
    \st \models \prop & \textif & \prop \in \stateval,\\
    \st \models \lnot\alpha & \textif & \st \not\models \alpha,\\
    \st \models \alpha_1 \land \alpha_2 & \textif & \st \models \alpha_1 \text{ and } \st \models \alpha_2,\\
    \st \models \expbel{i} \alpha & \textif & \alpha \in \belbaseset_i.
\end{eqnarray*}
\end{definition}
The explicit belief operator $ \expbel{i}$ has a
 set-theoretic interpretation:
agent $i$ has the explicit belief that $\alpha$
at  state $\st$
if and only if $\alpha$ is included in its  belief base
at $\st$. 

We define the notion 
of appetitive desire
(resp. aversive desire) from
explicit belief
combined with reward (resp. punishment).
Specifically,
we call agent $i$'s  appetitive desire base (resp. agent $i$'s aversive desire base)
at a  state $\st $
the set of all  facts
that, according to agent $i $'s explicit beliefs at $\st$,
entail a reward (resp. a punishment). 
The special atoms $\rewa{i} $
and $\punis{i} $ are key elements of the definition. 
\begin{definition}[Appetitive and aversive desire base]
\label{appaverdb}
Let $i \in \AGT$ and $\st \in \setbelbase$.
We define: 
    \begin{align*}
&\desbaseset{i}(\st  )=\{\alpha \in \fraglang \suchthat 
\alpha \rightarrow \rewa{i} \in \belbaseset_i
\},\\
&\undesbaseset{i}(\st  )=\{\alpha \in \fraglang \suchthat
\alpha \rightarrow \punis{i} \in \belbaseset_i
\}.
\end{align*}
$\desbaseset{i}(\st  )$
and $\undesbaseset{i}(\st  )$
are, respectively,
agent $i$'s appetitive and aversive
desire bases at state $\st$. 
\end{definition}
Note that we could have taken 
appetitive and aversive
desire bases
 as primitive concepts. We prefer to define them  from the 
 primitive concept of belief base
for the sake of minimality. 
The solution
we adopt uses only one primitive concept (belief base)
instead of three. This point  is reminiscent  
 of the debate  in philosophy
 opposing   the Humean to the non-Humean 
view of desire 
\cite{Karlsson2000}. 
 While according to the Humean view 
 desires and aversions
 are distinguished from beliefs, according
 to the non-Humean view,
 also called the
\emph{desire-as-belief (DAB) thesis},  
 they are reducible
 to them.\footnote{
 An  issue of debate 
 is the 
 violation  by the DAB thesis of  the 
  independence requirement
  between desire and belief change
  (see  \cite{DesireLewis1,DesireLewis2}). 
  In 
  \cite{BradleyList2008}
  it is shown  that under some  conditions
  of separation between evaluative and non-evaluative
  formulas, the requirement is not violated.
  We leave the  analysis of    analogous  conditions
 in our framework to future work. 
  } In this respect,
 our semantics is in line
 with the non-Humean view.

\subsection{Accessibility relations}

Three types
of accessibility relation, for epistemic possibility,
attraction and repulsion,
can be directly computed from
the agents' belief bases that, following Definition \ref{state}, are included 
in the state description.  
The following definition introduces  the 
epistemic accessibility relation. 
\begin{definition}[Epistemic alternatives]\label{DefAlternative}
Let $i \in \AGT$.
$\relstate{i}$
is the binary relation on  $\setbelbase$
such that,
for all
$ \st = \big( (\belbaseset_i)_{i \in \AGT}, \stateval \big),
\st '= \big( (\belbaseset_i')_{i \in \AGT}, \stateval' \big) \in \setbelbase$:
\begin{align*}
&\st \relstate{i} \st' \text{ if and only if } 
\forall \alpha \in \belbaseset_i:
\st' \models \alpha. 
\end{align*}
\end{definition}
$\st \relstate{i} \st '$ means
that
at state 
$\st$
agent $i$
considers state $\st'$
(epistemically) possible.
According to the definition,
the latter is the case 
if and only if
$\st'$ satisfies all
facts that agent $i$ explicitly believes
at $\st$. The set 
$\relstate{i}(\st)= \{ \st' \in \setbelbase : \st\relstate{i} \st ' \}$
is agent $i$'s epistemic state at $\st$.

The relation $\relstate{i}$ as defined above is not necessarily reflexive.
As explained in \cite{LoriniAI2020}, reflexivity can be obtained by restricting it to the subclass of ``belief correct'' states in which the information in an agent's belief base is actually true.
We do not make this restriction since we model belief instead of knowledge.
Unlike knowledge, an agent's belief can be incorrect (i.e., it might be the case that an agent believes that $\varphi$ is true whereas $\varphi$ is actually false).
We also do not consider introspection for belief.
However, there is a natural way to modify the accessibility relation $\relstate{i}$ to make implicit belief introspective.
It would be sufficient to add the condition that, for a state $\st'$ to be accessible from state $\st$, agent $i$'s belief base at state $\st'$ is identical to agent $i$'s belief base at state $\st$
(i.e., $\belbaseset_i=\belbaseset_i'$).
By adding this extra condition, the relation $\relstate{i}$ would become transitive and Euclidean.
Transitivity corresponds to positive introspection, while Euclidianity corresponds to negative introspection.

Analogously, we compute an agent's 
accessibility
relations for
attraction and repulsion
from its  appetitive  and aversive desire bases,
that as shown in Definition 
\ref{appaverdb} are defined from its  belief base. 
Specifically, 
we suppose that an agent is attracted to a state 
(or finds it attractive) if and only if at least one of its  appetitive desires is satisfied at  the state, i.e., the state makes something the agent wishes to achieve true. Conversely, we suppose that an agent is repelled  by a state (or finds it repulsive) if and only if at least one of its   aversive desires  is realized   at  the state, i.e., the state makes something the agent wishes to avoid true.
\begin{definition}[Attractive and repulsive  alternatives]\label{DefAttractiveRepulsiveAlternative}
Let $i \in \AGT$. 
$\relattr{i}$ and
 $\relrepu{i}$
are the binary relations on  $\setbelbase$
such that,
for all
$ \st = \big( (\belbaseset_i)_{i \in \AGT}, \stateval \big),
\st '= \big( (\belbaseset_i')_{i \in \AGT}, \stateval ' \big) \in \setbelbase$:
\begin{align*}
&\st \relattr{i} \st ' \text{ if and only if } 
\exists  \alpha \in \desbaseset{i}(\st ) \text{ s.t. }
\st' \models \alpha, \\
&\st \relrepu{i} \st' \text{ if and only if } 
\exists  \alpha \in \undesbaseset{i}(\st ) \text{ s.t. }
\st' \models \alpha.
\end{align*}
\end{definition}
$\st \relattr{i} \st'$ means
that
at state $\st$
agent $i $ is attracted to state  $\st'$,
whereas 
$\st \relrepu {i} \st'$ means
that
at $\st$ it  is repelled  by $\st'$.

\subsection{Models}
 
The following definition introduces the concept of model, namely a state supplemented with a set of states, called \emph{context}. 
The latter includes all states that are compatible with the agents' common ground,
i.e., the body of information that they  commonly believe to be the case \cite{StalnakerCommonGround}. 

\begin{definition}[Model]%
\label{univM}
A model is a pair $(\st,\iconstraint)$
such that  $\st\in\iconstraint\subseteq \setbelbase$. 
The class of models is noted $\classbelbase$.
\end{definition}

Note that we suppose  the actual state
to be included in the agents' common ground. 
However, validities would not change even if we did not suppose  this, 
due to the fact that we model belief instead of knowledge.
If we modeled knowledge instead of belief,  supposing  that $\st \in \iconstraint$
would  become a necessary requirement.  Note also that 
$\iconstraint$ 
can be any  (possibly infinite) set
of states with no restriction on the information
that  is included in belief bases. 
Nonetheless, in some cases it is   useful  to
restrict  to models in which an agent's belief base
is constructed
from a finite  repository of information,
or vocabulary. 
The  agent's vocabulary plays a role analogous to that of awareness 
in  \cite{Fag87}. 
This observation leads to the following definition. 
 
\begin{definition}[$\setrelevantformulas$-model]%
\label{univM2}
Let $\setrelevantformulas = (\setrelevantformulas_i)_{i \in \AGT}$
where, for every $i\in \AGT$, $\setrelevantformulas_i\subseteq \fraglang$
represents agent $i$'s vocabulary.
The model $(\st,\iconstraint)$
in $\classbelbase$
is said to be a $\setrelevantformulas$-model 
if $\st \in \iconstraint=\setbelbasecustom$,
with
$
\setbelbasecustom=\Big\{ 
  \big( (\belbaseset_i')_{i \in \agtset},\stateval'\big)
\in \setbelbase\suchthat \forall i \in \agtset, 
\belbaseset_i' \subseteq \setrelevantformulas_i
\Big\}.
$
The class of $\setrelevantformulas$-models
is noted $\classbelbaseuniv(\setrelevantformulas)$.
\end{definition}

$\setrelevantformulas = (\setrelevantformulas_i)_{i \in \AGT}$ 
in Definition \ref{univM2} is also  called agent vocabulary profile.
Clearly, when $\setrelevantformulas_i= \fraglang $
for every $i\in \agtset $,
we have  $\setbelbasecustom=\setbelbase$. 
A model  $(\st,\setbelbase)$  is a model with maximal ignorance:
it only  contains the information provided by the actual  state $\st$.

\section{Language}%
\label{sec:languaeg}

 We define a general   modal  language 
 $\lang$
 extending the  
 language $\fraglang$
by five types of modal operators for 
cognitive attitudes:
 implicit belief ($\impbel{i}$),
 complete attraction ($\attract{i}{ }$), complete repulsion ($\repuls{i}{ }$), realistic attraction
 ($  \attractreal{i}{}$)
 and realistic repulsion ($\repulsreal{i}{}$). 
 It 
is defined by the following grammar:
 \begin{center}\begin{tabular}{llcl}
 $\lang $ & $ \defin $  $\varphi$  & $\bnf$ & $\alpha  \mid \neg\varphi \mid \varphi  \wedge \varphi   \mid  \impbel{i}\varphi 
 \mid 
  \attract{i}{ } \varphi
  \mid \repuls{i}{ } \varphi \mid$ \\
  && & $
  \attractreal{i}{}\varphi  \mid \repulsreal{i}{}\varphi , 
                        $\
\end{tabular}\end{center} 
where $\alpha$ ranges over $\fraglang$
and $i $ ranges over $\AGT $.
The other Boolean constructions $\top$, $\bot$, $\lor$, $\limply$
and $\lequiv$ are defined from $\alpha$, $\neg$ and $\land$ as usual. 
The formula $ \impbel{i}\varphi$ is read
``agent $i$ implicitly believes that $\varphi$''.
The formula $ \attract{i}{}   \varphi$ is read
``agent $i$ is completely 
attracted to  the fact
that $\varphi$'',
whereas $ \repuls{i}{}   \varphi$ is read
``agent $i$ is 
completely 
repelled by the fact
that $\varphi$''. Furthermore, 
the formula $ \attractreal{i}{}   \varphi$ is read
``agent $i$ is realistically  attracted to the fact
that $\varphi$'',
whereas $ \repulsreal {i}{}   \varphi$ is read
``agent $i$ is realistically repelled by the fact
that $\varphi$''.
The satisfaction relation
between models 
and $\lang$-formulas
is defined  as follows. 
(Boolean cases are omitted since they are defined as usual.) 
\begin{definition}[Satisfaction relation, cont.]\label{truthcond4}
Let  $ (\st ,\iconstraint)   \in\classbelbase$. We define $(\st, \iconstraint) \models \varphi$ by:
\begin{eqnarray*}
    (\st,\iconstraint) \models \alpha & \textif & \st \models \alpha,\\
    (\st,\iconstraint)\models \impbel{i}\varphi & \textif
        & \forall \st' \in  \iconstraint , \text{ if } \st \relstate{i} \st'
            \text{ then }\\
&&             (\st',\iconstraint) \models \varphi,\\
             (\st,\iconstraint) \models \attract{i}{} \varphi & \textif & \forall \st' \in  \iconstraint  , \text{if }
  ( \st' , \iconstraint) \models \varphi \text{ then } \\
  && \st \relattr{i} \st',\\
 (\st ,\iconstraint) \models \repuls{i}{} \varphi & \textif & \forall \st' \in  \iconstraint  , \text{ if }
  ( \st' , \iconstraint) \models \varphi \text{ then }  \\
  && \st \relrepu{i} \st',\\
   (\st,\iconstraint) \models \attractreal{i}{}\varphi 
 & \textif & 
 \forall \st' \in  \iconstraint  , \text{ if }
  ( \st' , \iconstraint) \models \varphi
  \text{ and }\\
  &&   \st\relstate{i} \st'\text{ then }   \st\relattr{i} \st',\\
    (\st,\iconstraint) \models \repulsreal{i}{}\varphi
 & \textif &
 \forall \st' \in  \iconstraint , \text{ if }
  ( \st' , \iconstraint) \models \varphi
  \text{ and }
 \\
&&   \st\relstate{i}\st' \text{ then }   \st\relrepu{i} \st'.
\end{eqnarray*}
\end{definition}
Interpretations of the
different modalities 
are restricted to the actual context $\iconstraint$.  
The semantic interpretations
of the three modalities 
$\impbel{i}$, $\attract{i}{}$
and 
$\repuls{i}{}$
highlight their conceptual symmetry. 
The implicit belief
operator  is interpreted in the usual way:
agent $i$
implicitly believes
that $\varphi$, i.e. $\impbel{i} \varphi $, 
if and only if $\varphi
$
is true at all states that she considers epistemically  possible. 
 Conversely, agent $i$ is completely
 attracted to (resp. repelled by)  $\varphi$, i.e. 
 $\attract{i}{} \varphi $ (resp. $\repuls{i}{} \varphi$), 
 if and only if every state satisfying $\varphi$
 is  attractive (resp. repulsive)  to her. 
 Modalities $\attract{i}{}$
 and $\repuls{i}{}$
 are instances of the so-called
 ``window'' modality  \cite{Humberstone1983,GargovEtAl1987,Benthem79,Goranko1990}. The idea of using  a
 ``window'' modality for representing desire
 was  defended in \cite{DBLP:journals/mima/DuboisLP17}.
 Realistic attraction (resp. repulsion)  is an agent's attraction
 (resp. repulsion) relative
 to its  epistemic state:
  agent $i$ is realistically attracted 
 to 
 (resp. repelled by) 
 $\varphi$, i.e.
 $\attractreal {i}{}\varphi$ (resp. $\repulsreal{i}{}\varphi$),
 if every $\varphi$-state
in its  epistemic state is attractive (resp. repulsive).  
 
Notions of satisfiability and validity of $\lang$-formulas for the class of models $\classbelbase$ are defined in the usual way:
$\varphi$ is satisfiable if there exists $(\st,\iconstraint) \in \classbelbase$ such that $(\st,\iconstraint) \models \varphi$,
and
$\varphi$ is valid (noted $\models \varphi $) if $\lnot\varphi$ is not satisfiable. 
 We now establish that the operators $\attract{i}{} $, $\repuls{i}{} $, $\attractreal{i}{}$ and $ \repulsreal{i}{}$ are not definable using the other operators, including each other.

\begin{theorem}\label{theo:unexpr}
     The operators
     $\attract{i}{} $,
     $\repuls{i}{} $,
      $\attractreal{i}{}  $
     and 
     $ \repulsreal{i}{} $
     are not expressible
     with the other modalities or each other.
\end{theorem}

This result still holds
in the context of the dynamic extension
we will present later in the paper.

\section{Axiomatics}\label{sec:axioms}

 In this section we provide  a sound and complete axiomatization for the class of models $\classbelbase$.

\begin{definition}
 We define the logic $\mathsf{LCA}$
 (Logic of Cognitive Attitudes)  to be the
extension of classical propositional logic given by the following axioms and rules
of inference:
\begin{align}
& \big( \impbel{i}\varphi \wedge  \impbel{i}(\varphi \rightarrow \psi)  \big) 
\rightarrow \impbel{i}\psi  \tagLabel{A1}{ax:A1}\\
& \big( \circledcirc    \varphi \wedge \circledcirc (\neg \varphi \wedge \psi)  \big) 
\rightarrow \circledcirc  \psi \tagLabel{A2}{ax:A2}\\
& \expbel{i }\alpha \rightarrow \impbel{i} \alpha
\tagLabel{A3}{ax:A3}\\
& \expbel{i }(\alpha\rightarrow \rewa{i} )\rightarrow \attract{i}{} \alpha \tagLabel{A4}{ax:A4} \\
& \expbel{i }(\alpha\rightarrow \punis{i} )\rightarrow \repuls{i}{} \alpha \tagLabel{A5}{ax:A5} \\
    &\attract{i}{}\varphi \rightarrow   \attractreal {i}{}\varphi \tagLabel{A6}{ax:A6} \\
        &\repuls{i}{}\varphi \rightarrow   \repulsreal  {i}{}\varphi \tagLabel{A7}{ax:A7}\\
   & \impbel{i}{}\varphi \rightarrow   \attractreal{i}{} \neg \varphi \tagLabel{A8}{ax:A8}\\
      & \impbel{i}{}\varphi \rightarrow   \repulsreal {i}{} \neg \varphi \tagLabel{A9}{ax:A9}\\
    &\attractreal{i}{}\varphi \rightarrow  \impbel{i} (\varphi \rightarrow \rewa{i}) \tagLabel{A10}{ax:A10}\\
    & \repulsreal{i}{}\varphi \rightarrow  \impbel{i} (\varphi \rightarrow \punis{i}) \tagLabel{A11}{ax:A11}\\
    & \frac{ \varphi }{ \impbel{i}\varphi } \tagLabel{R1}{ax:R1}\\
     & \frac{ \varphi }{  \circledcirc  \neg \varphi }
     \tagLabel{R2}{ax:R2} \\
     & \frac{ \varphi\lequiv\psi }{ \circledcirc\varphi\lequiv\circledcirc\psi }
     \tagLabel{R3}{ax:R3}
\end{align}
for $\circledcirc \in \{    \attract{i}{},\repuls {i}{} , 
\attractreal {i}{},\repulsreal  {i}{} \}$.%
\end{definition}

Most of these axioms and rules are natural to consider and directly follow from the semantics of the operators.
Let us highlight the four  that may not be immediately intuitive: 
\ref{ax:A2},
\ref{ax:A8},
\ref{ax:A9}, and \ref{ax:R2}. 
\ref{ax:A2} can be interpreted as follows.
Suppose an agent is completely attracted to the fact that $\varphi$
(i.e., $\attract{i}{}\varphi$). 
This means it  finds  attractive all situations at which 
$\varphi$ holds, 
which implies that it finds attractive all situations at which
$\varphi \wedge \psi $ holds.
 Moreover, suppose
the agent is completely attracted to the fact that  $\neg \varphi \land  \psi$
(i.e., $\attract{i}{}(\neg \varphi \land  \psi) $).
This means that it finds attractive all situations at which 
$\neg \varphi \land  \psi$ holds.
But, finding attractive all situations at which 
$\varphi \wedge \psi $ holds
and 
finding attractive all situations at which 
$\neg \varphi \wedge \psi $ holds clearly
implies 
finding attractive all situations at which 
$\psi $ holds
(since the union of the set of 
$\varphi \wedge \psi $-situations
and    the set of 
$\neg \varphi \wedge \psi $-situations
equals
the set of $\psi $-situations). 
The latter means being completely attracted
to the fact that $\psi$
(i.e., $\attract{i}{}\psi$).
The corresponding versions of \ref{ax:A2} for complete repulsion, realistic attraction and realistic repulsion are justified in an analogous way. 
\ref{ax:A8},
\ref{ax:A9}, and \ref{ax:R2}
 say that if $\varphi$ is believed then its negation is both realistically attractive and realistically repulsive, and that the negation of a valid fact is both completely and realistically attractive as well as repulsive. This can be understood as follows: if $\varphi$ is true in all considered situations, then all considered situations in which it is false (of which there are none) are both attractive and repulsive.  Put differently, 
when an agent believes that $\varphi$, she  has no realistic concern about $\neg \varphi$. Thus, she is willing to realistically  appreciate and despise $\neg \varphi$
 at the same time.

Checking soundness of this logic w.r.t.\ our class of models is straightforward.
The proof of the
completeness
follows the general
lines of the proofs  of
Theorem 1, Theorem 2 and Corollary 1
in 
\cite{LoriniAI2020}.
\begin{theorem}\label{theo:2}
    The logic $\mathsf{LCA}$
    is sound and complete
    for the class $\classbelbase$. 
\end{theorem}

\section{Cognitive positions}\label{sec:cognitivePositions}

It is interesting to study the 
language $\lang $ from a combinatorial perspective.
Specifically, it is worth 
combining  the four primitive  modalities 
$\attract{i}{} ,\repuls{i}{} ,\attractreal{i}{} $
and $\repulsreal{i}{} $
so as to define eight  cognitive
positions
of motivational type. 
In 
Table \ref{tab:cogatt}
we define the  four  notions  of
``being motivated by $\varphi$'' ($\motiv{i}\varphi$),
``being demotivated by $\varphi$'' ($\demotiv{i}\varphi$),
``being indifferent  about $\varphi$'' ($\indiff{i}\varphi$)
and 
``being ambivalent  about $\varphi$'' ($\ambi  {i}\varphi$).
The definiendum (e.g., $\motiv{i}\varphi$) is
defined by the conjunction 
of the two definiens (e.g., $\attract{i}{} \varphi \wedge
\neg \repuls {i}{} \varphi$). 
\begin{table}[h]
\centering
\begin{tabular}{c|c|c|}
  $\defin $   & $  \attract{i}{} \varphi$ & $ \neg  \attract{i}{} \varphi$ \\
    \hline
    $  \repuls{i}{} \varphi$ & $ \ambi{i}{} \varphi $ & $\demotiv{i}\varphi$ \\
    \hline
    $  \neg \repuls{i}{} \varphi$ & $\motiv{i}\varphi$ & $ \indiff{i}{} \varphi $ \\
    \hline
\end{tabular}
\caption{Cognitive attitudes  }
\label{tab:cogatt}
\end{table}

In 
Table \ref{tab:rcogatt},
we define 
their realistic counterparts: 
``being realistically  motivated by $\varphi$'' ($\realmotiv{i}\varphi$),
``being realistically  demotivated by $\varphi$'' ($\realdemotiv{i}\varphi$),
``being realistically
indifferent  about $\varphi$'' ($\rindiff{i}\varphi$)
and 
``being realistically
ambivalent  about $\varphi$'' ($\rambi  {i}\varphi$).

\begin{table}[h]
\centering
\begin{tabular}{c|c|c|}
   $\defin $  & $  \attractreal{i}{} \varphi$ & $ \neg  \attractreal{i}{} \varphi$ \\
    \hline
    $  \repulsreal{i}{} \varphi$ & $ \rambi {i}{} \varphi $ & $\realdemotiv{i}\varphi$ \\
    \hline
    $  \neg \repulsreal{i}{} \varphi$ & $\realmotiv{i}\varphi$ & $ \rindiff{i}{} \varphi $ \\
    \hline
\end{tabular}
\caption{Realistic cognitive attitudes  }
\label{tab:rcogatt}
\end{table}

From the definitions, it is easy to prove that 
being motivated
(resp. realistically motivated) 
implies not being demotivated
(resp. realistically demotivated) . 
\begin{proposition}\label{prop:motdemot}
Let $i\in \AGT$. Then,
\begin{align}
    & \models\motiv{i}\varphi \rightarrow \neg \demotiv{i}\varphi\label{motdemot1}, \\
    & \models\realmotiv{i}\varphi \rightarrow \neg \realdemotiv{i}\varphi. \label{motdemot2}
\end{align}
\end{proposition}

We leverage the  notions of being motivated/demotivated and their realistic variants 
to define two notions of dyadic preference between facts:
\begin{align*}
 \psi  \prec_i  \varphi& \defin 
(\motiv{i}\varphi  \wedge  \neg \motiv{i}\psi ) \vee 
(\demotiv{i}\psi   \wedge  \neg \demotiv{i}\varphi  ) , \\
 \psi  \prec_i^{\mathsf{real}}  \varphi & \defin 
(\realmotiv{i}\varphi  \wedge  \neg \realmotiv{i}\psi ) \vee 
(\realdemotiv{i}\psi   \wedge  \neg \realdemotiv{i}\varphi  ) . 
\end{align*}
The abbreviation $\psi  \prec_i  \varphi$
is read ``agent $i$
prefers $\varphi $
to $\psi$'', while 
 $ \psi  \prec_i^{\mathsf{real}}  \varphi$
is read ``agent $i$
realistically prefers $\varphi $ to $\psi$''. 
An agent prefers $\varphi$
to $\psi$
if
either i) she  is motivated 
by $\varphi$
without being motivated
by $\psi$, 
or  ii) she  is demotivated
by $\psi$
without being demotivated
by $\varphi$. 
The realistic preference
relation is defined analogously. 
It is straightforward to show
that both preference relations are strict partial orders.
\begin{proposition}\label{prefprop}
    Let $\jokertwo \in \{\prec_i ,  \prec_i^{\mathsf{real}} \}$. Then,
    \begin{align}
        &\models  \neg  (\varphi   \jokertwo  \varphi )\label{propirr},\\
         &  \models  (\psi   \jokertwo  \varphi ) \rightarrow \neg  (\varphi   \jokertwo  \psi) \label{propasy},\\
          &\models   \big( (\varphi_1   \jokertwo  \varphi_2 )
          \wedge (\varphi_2   \jokertwo  \varphi_3 )\big) 
          \rightarrow (\varphi_1   \jokertwo \varphi_3 ) \label{proptrans}.
    \end{align}
\end{proposition}
Validities (\ref{propirr}),
(\ref{propasy})
and 
(\ref{proptrans}) 
capture
irreflexivity,
asymmetry
and transitivity for preference, respectively. 
    
\section{Dynamic extension}\label{sec:dynext}

In this section we move from a static to a dynamic perspective on cognitive attitudes 
by presenting an extension of the language 
 $\lang $
 that supports reasoning about 
 the consequences
 of belief
change operations.
We see the latter as atomic  programs 
on an agent's cognitive state and use the constructs of dynamic logic for representing complex programs. 
Our  extension is defined by the following
grammar:
 \begin{center}\begin{tabular}{llcl}
  $\langprog $ & $ \defin $  $\pi $  & $\bnf$ & $ +_i \alpha \mid -_i \alpha
  \mid \pi ; \pi \mid \pi \cup \pi \mid ? \varphi ,
                        $\\
 $\langdyn $ & $ \defin $  $\varphi$  & $\bnf$ & $\alpha  \mid \neg\varphi \mid \varphi  \wedge \varphi   \mid  \impbel{i}\varphi 
\mid
  \attract{i}{ } \varphi
  \mid$ \\
  && & $\repuls{i}{ } \varphi
  \mid \attractreal{i}{}\varphi  \mid \repulsreal{i}{}\varphi
 \mid [\pi  ]\varphi,
                        $\
\end{tabular}\end{center} 
where $\alpha$ ranges over $\fraglang$
and $i $ ranges over $\AGT $.
 $\langprog $
is the language of programs. It includes atomic
programs for private belief expansion ($+_i \alpha $) and private forgetting ($-_i \alpha$)
as well as complex programs
of sequential composition ($;$),
non-deterministic choice ($\cup$) and test ($? $). 
The formula $[\pi  ]\varphi$
is read ``$\varphi$ holds after  program $\pi$ has been executed''.
The dual of the operator $[\pi  ]$ is defined in the usual way:
$\langle \pi  \rangle \varphi \defin \neg[\pi  ] \neg\varphi$.
The formula $\langle \pi  \rangle \varphi $
is read 
``there exists
an execution of the program
$\pi$ at the end of which $\varphi$ holds''.

To be able to interpret the language $\langdyn$
we extend the satisfaction relation of
Definition \ref{truthcond4}  as follows.
 \begin{definition}[Satisfaction relation, cont.]\label{truthconddynam}
Let  $ (\st ,\iconstraint)   \in\classbelbase$. We define:
\begin{eqnarray*}
             (\st ,\iconstraint) \models [\pi ] \varphi & \textif & 
\forall \st'  \in\iconstraint,
 \text{ if }
     \st  \reldyn{\pi  }^{\iconstraint}  \st' 
 \text{ then }\\
 &&
  (\st' ,\iconstraint)\models\varphi ; \text{ with }\\
    \st \reldyn{+_i \alpha }^{\iconstraint} \st' &\text{ iff } &
    \! \stateval=\stateval',
     \belbaseset_i^{+_i \alpha}= \belbaseset_i\cup \{\alpha\}
  \text{ and }\\
  && \forall j \neq i, 
  \belbaseset_j^{+_i \alpha}= \belbaseset_j,\\
       \st \reldyn{-_i \alpha }^{\iconstraint} \st'  &\text{ iff }&
    \! \stateval=\stateval',
     \belbaseset_i^{-_i \alpha}= \belbaseset_i\setminus \{\alpha\}
  \text{ and }\\
  && \forall j \neq i, 
  \belbaseset_j^{-_i \alpha}= \belbaseset_j,\\
         \st \reldyn{\pi_1 ; \pi_2 }^{\iconstraint} \st' & \text{ iff }&
    \exists \st''\in \iconstraint 
    \text{ such that  } 
    \st \reldyn{\pi_1  }^{\iconstraint} \st''\text{    and }\\
&&    \st'' \reldyn{\pi_2 }^{\iconstraint} \st' ,\\
         \st \reldyn{\pi_1 \cup \pi_2 }^{\iconstraint} \st' &\text{ iff }&
  \st \reldyn{\pi_1 }^{\iconstraint} \st'  \text{    or }
 \st \reldyn{ \pi_2  }^{\iconstraint} \st',\\
 \st \reldyn{?\varphi }^{\iconstraint} \st' &\text{ iff }& \st'=\st \text{ and }  (\st ,\iconstraint)\sat \varphi . 
\end{eqnarray*}
\end{definition}
 The dynamic modality  $[\pi ] $
 is interpreted in the expected way:
 $\varphi$ holds
 after program $\pi$
 is executed if
 $\varphi$
 holds at every state
 which is accessible
from the actual
one by executing program $\pi$. 
The atomic program $+_i \alpha$ 
for private belief expansion 
expands agent $i$'s belief
base with the formula $\alpha $,
while $-_i \alpha$
for private forgetting 
removes formula $\alpha $
from agent $i$'s belief
base. 
They are private operations since they keep
the belief bases of all agents different from $i$ unchanged. 
Sequential composition, non-deterministic choice
and test are interpreted in the usual way. 
Note that the  binary  relation $\reldyn{\pi }^{\iconstraint}$
is 
parameterized with  $\iconstraint$
to recall the context with respect
to which the formulas have to be evaluated. 
Let us illustrate the dynamic language  $\langdyn $.
with an  example. 

\begin{example}
\label{example:mother}
  A mother enters the room
of her child   Bob
after hearing a loud noise coming from there.
She sees that there is a big mess in the room. 
The mother's goal 
is  to motivate Bob  to tidy up the room. 
She has to choose the right combination
of speech acts to achieve her goal.
We suppose the mother's  repertoire of speech acts
includes the following three speech acts: 
SA1: ``You won't be allowed to watch TV this afternoon,
if you do not tidy up your room!'';
SA2: ``I'll prepare some good chocolate crepes for you, 
if you tidy up your
room!''; 
SA3: ``Don't worry, you won't necessarily get tired from tidying up the room, 
it will only take a few minutes!''.   Moreover, we suppose 
Bob has the following information in his belief base: i)
tidying up the room ($  \mathit{td}_{\mathit{Bob }}$)  is a tiring activity
($  \mathit{ti}_{\mathit{Bob }}$),
ii)  being tired is
a bad thing,
iii)  a crepe-based snack
($  \mathit{cr}_{\mathit{Bob }}$)
is a good thing, 
iv) 
being deprived of TV ($\neg  \mathit{tv}_{\mathit{Bob }}$) is a bad thing. 
In formal terms, let
$\AGT= \{ \mathit{Bob}\}$. 
The initial state is  $\st_0 = \big(  
  \belbaseset_{\mathit{Bob}}, \stateval_0 \big)$ with 
  \begin{align*}
       \belbaseset_{\mathit{Bob}}=&
      \{ 
      \mathit{td}_{\mathit{Bob} }
      \rightarrow    \mathit{ti}_{\mathit{Bob} }, 
   \mathit{ti}_{\mathit{Bob} }  \rightarrow \punis{\mathit{Bob}},  \\
   & \mathit{cr}_{\mathit{Bob} }
      \rightarrow    \rewa{\mathit{Bob} }, 
      \neg \mathit{tv}_{\mathit{Bob} } 
      \rightarrow    \punis{\mathit{Bob} } 
      \} 
  \end{align*}
  and 
  $ \stateval_0= \emptyset $. 
The speech acts SA1, SA2 and SA3 are represented  formally
by the atomic programs 
$+_{\mathit{Bob}}   \alpha_1$,
$+_{\mathit{Bob}}   \alpha_2 $ and 
 $ -_{\mathit{Bob}}   \alpha_3  $
 with 
$ \alpha_{1}  \defin 
\neg \mathit{td}_{\mathit{Bob}} 
  \rightarrow 
   \neg \mathit{tv}_{\mathit{Bob}} $,
   $
    \alpha_{2}  \defin 
 \mathit{td}_{\mathit{Bob}} 
  \rightarrow 
    \mathit{cr}_{\mathit{Bob}} $ and 
$    \alpha_{3}  \defin 
 \mathit{td}_{\mathit{Bob}}
  \rightarrow 
   \mathit{ti}_{\mathit{Bob}}$. 
   It is routine  to verify that:
  \begin{align}
   (\st_0,\setbelbasecustom   ) \models  ( \mathit{td}_{\mathit{Bob}}
 \prec_{\mathit{Bob}}^{\mathsf{real}} \neg  \mathit{td}_{\mathit{Bob}}), \label{check1}
  \end{align}
with 
$\Gamma=( 
\belbaseset_{\mathit{Bob}} )$. 
  This means that at the initial state  $\st_0$
  Bob realistically prefers 
  not tidying up his room
  to tidying it up. Moreover, 
it is routine  to verify that:
  \begin{align}
   (\st_0,\setbelbasecustom) \models 
  [ \pi_{\mathit{talk}}]    (\neg \mathit{td}_{\mathit{Bob}}
 \prec_{\mathit{Bob}}^{\mathsf{real}}  \mathit{td}_{\mathit{Bob}}), 
  \end{align}
with 
 $\pi_{\mathit{talk}} \defin  \bigcup_{\epsilon, \epsilon' \in \{+_{\mathit{Bob}}   \alpha_1,
  +_{\mathit{Bob}}   \alpha_2, -_{\mathit{Bob}}   \alpha_3\}
  \suchthat \epsilon\neq \epsilon' 
  }
  \epsilon; \epsilon' $.
  This means any sequence of two different speech acts
  from the mother's speech act repertoire 
  is sufficient to reverse  Bob's
    realistic preference
    and, consequently, to make
    the mother achieve her goal. 
  
Note that
the example can be generalized  to an arbitrary
set of  children  $\AGT=\{1, \ldots , n\}$
by  replacing the information
``tidying up the room is a tiring activity'' by the information
``tidying up the room
\emph{without the help of the others} is a
tiring activity'' 
in the belief base of each child
and keeping everything else the same
as in the single-agent case. 
In particular, 
in the general case
we should consider the initial state  $\st_0 = \big( (\belbaseset_i)_{i \in \AGT } , \stateval_0 \big)$
such that $
    \belbaseset_i= 
      \big\{ 
     (   \bigwedge_{ j \neq i } \neg \mathit{td}_{j } \land \mathit{td}_{i} )
      \rightarrow    \mathit{ti}_{i},
      \mathit{ti}_{i} \rightarrow \punis{i}, 
           \mathit{cr}_{i}
      \rightarrow    \rewa{i }, 
      \neg \mathit{tv}_{i}
      \rightarrow    \punis{i } 
      \big\}$ 
  and 
  $ \stateval_0= \emptyset $. 
  Then,
we should check the following
statement instead of  the statement 
(\ref{check1}) given above: 
  \begin{align}
   (\st_0,\setbelbasecustom   ) \models 
   \bigwedge_{i \in \AGT } 
\big(  & (  \bigwedge_{ j \neq i } \neg \mathit{td}_{j } \land \mathit{td}_{i}  ) 
 \prec_{i }^{\mathsf{real}} \notag   
 \\
&  ( \mathit{td}_{i}\rightarrow  \bigvee_{ j \neq i }  \mathit{td}_{j }   ) \big) , \label{check3}
  \end{align} 
with 
$\setrelevantformulas = (\belbaseset_i )_{i \in \AGT}$. 
The formula to be checked expresses the fact that at the initial state every child
realistically prefers 
having someone else helping them to tidy up the room
to tidying up the room on their own. 
\end{example}

Note that   
the set $\setbelbasecustom  $
in the  Example \ref{example:mother}
 corresponds to the state space in a Kripke model used to represent the scenario.
 It is clearly exponential in the size of $\Gamma$. But
 $\setbelbasecustom  $
 and  Bob's accessibility relations do not need to be explicitly represented in the formulation of the
 model checking problem. We only need to represent the initial state and the set $\Gamma$. So, we can represent the example in an exponentially more succinct way than in the traditional Kripke semantics
 of epistemic logic.

 The next section is devoted to 
exploring the model checking 
problem for the language $\langdyn$
from a
complexity and 
algorithmic perspective. 
The generalization
of the example to an arbitrary set of agents
will turn out to be useful since we will use it to
test
the performance of our model checking algorithm 
as a function of the  number
of children.

\section{Model checking}%
\label{sec:mc}
\newcommand\mc[3]{mc(#1, #2, #3)}

The model checking problem for $\langdyn$ is defined as follows: given 
    an agent vocabulary profile $\setrelevantformulas = (\setrelevantformulas_i)_{i \in \AGT}$ with $\setrelevantformulas_i$ finite for every $i\in \AGT$,
    a finite state $\st_0$ in $\setbelbasecustom$,
    and
    a formula $\varphi_0 \in \langdyn$, decide whether $(\st_0,\setbelbasecustom) \models \varphi_0$ 
    holds.

\begin{theorem}
    The model checking problem for $\langdyn$ is  PSPACE-complete.
\end{theorem} 

PSPACE-hardness comes from the PSPACE-hardness of the
model checking for the single-agent fragment 
of $\lang $ with only the
implicit belief operator $\impbel {i } $ 
\cite{DBLP:journals/corr/abs-1907-09114}.
PSPACE-membership comes from the poly-time reduction into TQBF (true quantified binary formulas) that follows.

\subsection{Reduction into TQBF}
\label{subsection-reductiontoTQBF}

\newcommand{\propfortriangle}[2]{x_{#1, #2}}
\newcommand{\setQBFvariableslevel}[1]{X_{#1}}

\newcommand{\QBFrel}{E}
\newcommand\trrel[2]{tr_{#1, #2}^{\text{prog}}}
\newcommand{\QBFsymbol}{s}
\newcommand{\QBFfreshsymbol}{{s''}}
\newcommand{\QBFsymbolb}{{s'}}

We use 
abstract symbols
$\QBFsymbol$ (the reader may think of them as integers) to denote \emph{depths} in the translation of a $\lang$-formula into a QBF-formula.
We introduce TQBF propositional variables $\propfortriangle \alpha \QBFsymbol$ for all $\alpha \in \fraglang$ and for all symbols $\QBFsymbol$.
The variables indexed by $\QBFsymbol$ are said to be of level~$\QBFsymbol$.

Let $\setQBFvariableslevel {\QBFsymbol}$ be the set containing exactly 
all formulas 
$\propfortriangle {\expbel i \alpha} \QBFsymbol$
$\propfortriangle {\expbel i (\alpha \to \rewa i)} \QBFsymbol$
and $\propfortriangle {\expbel i (\alpha \to \punis i)} \QBFsymbol$
with $\alpha \in \setrelevantformulas_i$ for any agent~$i$,
and all $\propfortriangle p \QBFsymbol$ with $p$ appearing in $\setrelevantformulas$ or~$\supervarphi$.

Let us give some macros that will be used in the reduction.

\newcommand{\sameprop}[2]{\text{eq}_{\text{prop}}(#1, #2)}
\newcommand{\samej}[2]{\text{eq}_j(#1, #2)}
\newcommand{\sameplusi}[3]{\text{eq}_i^{+#3}(#1, #2)}
\newcommand{\sameminusi}[3]{\text{eq}_i^{-#3}(#1, #2)}

\begin{itemize}
    \item The valuations of the two states are equal:
    $$\sameprop\QBFsymbol\QBFsymbolb := \lbigand_p (\propfortriangle \prop \QBFsymbol \lequiv \propfortriangle \prop \QBFsymbolb).$$
    \item Agent $j$'s  belief bases in the two states are equal:
    $$\samej\QBFsymbol\QBFsymbolb := \lbigand_{\alpha \in \Gamma_j} (\propfortriangle {\expbel{j}\alpha} \QBFsymbol \lequiv \propfortriangle {\expbel{j}\alpha} \QBFsymbolb).$$
    \item Agent $i$'s  belief bases in the two states are equal, except for $\alpha$ which is added in the second state:
\begin{align*}
   \sameplusi\QBFsymbol \QBFsymbolb\alpha :=  x_{\expbel{i}\alpha, s'}  \land  
\lbigand_{\beta \in \Gamma_i \suchthat \beta \neq \alpha} \!\!\! (\propfortriangle {\expbel{i}\beta} \QBFsymbol \lequiv \propfortriangle {\expbel{i}\beta} \QBFsymbolb).
\end{align*}
    \item Agent $i$'s belief bases in the two states are equal, except for $\alpha$ which is removed from the second state:
\begin{align*}
\sameminusi\QBFsymbol \QBFsymbolb\alpha :=  \lnot x_{\expbel{i}\alpha, s'}  \land  \lbigand_{\beta \in \Gamma_i \suchthat \beta \neq \alpha} \!\!\!\! (\propfortriangle {\expbel{i}\beta} \QBFsymbol \lequiv \propfortriangle {\expbel{i}\beta} \QBFsymbolb).
\end{align*}
\end{itemize}

We now give the translation from $\lang$ to QBF. In this definition, we write $\forall X_s$ for the sequence $\forall x_1 \dots \forall x_m$ where $x_1, \dots, x_m$ is any\footnote{The semantics does not depend on the particular order.} enumeration of the variables in $X_s$.

\begin{definition}
\label{definition:reductiontoTQBF}
For all state symbols $s, s'$, we define a function $tr_s$ that maps any formula of 
$\lang$ to a QBF-formula, and a function 
$\trrel{s}{s'}$ that maps any program $\pi$ 
to a a QBF-formula by mutual induction as follows:
\fbox{
\begin{minipage}{0.45\textwidth}%
\vspace*{-3mm}
\begin{align*}
    tr_\QBFsymbol(\prop) & := \propfortriangle \prop \QBFsymbol \\
   tr_\QBFsymbol(\lnot\varphi) & := \lnot tr_\QBFsymbol(\varphi) \\
    tr_\QBFsymbol(\varphi \land \psi) & := tr_\QBFsymbol(\varphi) \land tr_\QBFsymbol(\psi) \\
        tr_\QBFsymbol(\expbel{i}\alpha) & :=
    \begin{cases}
        \propfortriangle{\expbel{i}\alpha}{\QBFsymbol}, & \text{ if } \alpha \in \Gamma_i\\
        \bot, & \text{ otherwise}
    \end{cases}\\
    tr_\QBFsymbol(\impbel{i}\varphi) & := \forall \setQBFvariableslevel{\QBFfreshsymbol}(\QBFrel_{i,\QBFsymbol, \QBFfreshsymbol} \limply tr_{\QBFfreshsymbol}(\varphi)) \\
    tr_\QBFsymbol(\attract{i}{} \varphi) & := \forall \setQBFvariableslevel{\QBFfreshsymbol}(tr_\QBFfreshsymbol(\varphi) \limply A_{i,\QBFsymbol, \QBFfreshsymbol} ) \\
    tr_\QBFsymbol(\repuls{i}{} \varphi) & := \forall \setQBFvariableslevel{\QBFfreshsymbol}(tr_\QBFfreshsymbol(\varphi) \limply R_{i,\QBFsymbol, \QBFfreshsymbol} ) \\
    tr_\QBFsymbol(\attractreal{i}{} \varphi) & := \forall \setQBFvariableslevel{\QBFfreshsymbol}(tr_\QBFfreshsymbol(\varphi) \land E_{i,\QBFsymbol, \QBFfreshsymbol}  \limply A_{i,\QBFsymbol, \QBFfreshsymbol} ) \\
    tr_\QBFsymbol(\repulsreal{i}{} \varphi) & := \forall \setQBFvariableslevel{\QBFfreshsymbol}(tr_\QBFfreshsymbol(\varphi) \land E_{i,\QBFsymbol, \QBFfreshsymbol}  \limply R_{i,\QBFsymbol, \QBFfreshsymbol} ) \\
    tr_\QBFsymbol([\pi]\varphi) & := \forall X_{\QBFfreshsymbol} (\trrel \QBFsymbol {\QBFfreshsymbol} (\pi) \limply tr_{\QBFfreshsymbol}(\varphi) )
\end{align*}
\end{minipage}
}

with 
\begin{align*}
    \QBFrel_{i,\QBFsymbol, \QBFfreshsymbol} & := \bigwedge_{\alpha \in \setrelevantformulas_i} \propfortriangle{\expbel i \alpha}\QBFsymbol  \limply tr_{\QBFfreshsymbol}(\alpha),
    \\
    A_{i,\QBFsymbol, \QBFfreshsymbol} & :=
    \bigor_{(\alpha \rightarrow \rewa{i}) \in \setrelevantformulas_i} 
    \propfortriangle{\expbel i (\alpha \rightarrow \rewa{i})}\QBFsymbol  \land tr_{\QBFfreshsymbol}(\alpha),
    \\    
    R_{i,\QBFsymbol, \QBFfreshsymbol} & := 
    \bigor_{(\alpha \rightarrow \punis{i}) \in \setrelevantformulas_i} 
    \propfortriangle{\expbel i (\alpha \rightarrow \punis{i})}\QBFsymbol  \land tr_{\QBFfreshsymbol}(\alpha),
\end{align*}
and $\trrel{s}{s'}$ is given by

\noindent
 \fbox{
\begin{minipage}{0.45\textwidth}%
\vspace*{-2mm}
\begin{align*}
    \trrel  \QBFsymbol \QBFsymbolb (+_i \alpha)  := &   \bigwedge_{j \neq i} \samej s {s'}  \\
    &  \land \sameplusi s {s'}\alpha \land \sameprop s {s'} \\
      \trrel \QBFsymbol \QBFsymbolb (-_i \alpha)  := &  \bigwedge_{j \neq i} \samej s {s'} \\
      &  
      \land \sameminusi s {s'}\alpha \land \sameprop s {s'} \\
  \trrel \QBFsymbol \QBFsymbolb (\pi_1; \pi_2)  := & \exists X_{\QBFfreshsymbol} (\trrel \QBFsymbol \QBFfreshsymbol (\pi_1) \land \trrel {\QBFfreshsymbol} {\QBFsymbolb}(\pi_2)) \\
  \trrel \QBFsymbol \QBFsymbolb (\pi_1 \union \pi_2)  := &    \trrel \QBFsymbol \QBFsymbolb (\pi_1) \lor   \trrel \QBFsymbol \QBFsymbolb (\pi_2) \\
  \trrel \QBFsymbol \QBFsymbolb (?\varphi)  := & \bigwedge_j  \samej\QBFsymbol\QBFsymbolb \land  
  \sameprop \QBFsymbol\QBFsymbolb \land tr_\QBFsymbol(\varphi)
\end{align*}
\end{minipage}
}

where $\QBFfreshsymbol$ in the clauses $tr_\QBFsymbol(\impbel{i}\varphi)$, $tr_\QBFsymbol(\attract{i}{}\varphi)$, $tr_\QBFsymbol(\repuls{i}{}\varphi)$, $tr_\QBFsymbol(\attractreal{i}{}\varphi)$, $tr_\QBFsymbol(\repulsreal{i}{}\varphi)$,
$tr_\QBFsymbol([\pi]\varphi)$ and $\trrel \QBFsymbol \QBFsymbolb (\pi_1; \pi_2)$ is each time a fresh abstract symbol (fresh means that this symbol was not used before).
\end{definition}

State $S$ (resp. $S'$) is represented by (the truth values of variables in) $X_\QBFsymbol$ (resp. $X_{\QBFfreshsymbol}$). Given a state symbol $s$, given a formula $\varphi$ of $\langdyn$, the intended meaning of $tr_s(\varphi)$ is that formula $\varphi$ holds in the state $S$ represented by the truth values of variables in $X_\QBFsymbol$.
The translation of $\expbel{i}\alpha$ is $\bot$ when $\alpha$ is not in the corresponding belief base.
The translation $tr_{\QBFsymbol}(\impbel{i}\varphi)$ directly reflects the semantics of $\impbel{i}\varphi$; the same for the other operators $\attract{i}{}$, $\repuls{i}{}$, $\attractreal{i}{}$, $\repulsreal{i}{}$.
Formula $\QBFrel_{i,\QBFsymbol,\QBFfreshsymbol}$ reformulates  $\st \relstate i \st'$, and similarly $A_{i,\QBFsymbol, \QBFfreshsymbol}$ and $R_{i,\QBFsymbol, \QBFfreshsymbol}$ reformulate $\st \relattr i \st'$ and $\st \relrepu i \st'$ respectively.

For the translation $tr_\QBFsymbol([\pi]\varphi)$, we use $\trrel \QBFsymbol {\QBFfreshsymbol} (\pi)$ which is true iff $ \st  \reldyn{\pi  }^{\iconstraint}  \st' $ where $S$ and $S'$ are respectively represented by $X_\QBFsymbol$ and $X_\QBFsymbolb$. The definition of $\trrel \QBFsymbol {\QBFfreshsymbol} (\pi)$  reflects Definition~\ref{truthconddynam}.

\newcommand{\descr}[2]{\operatorname{desc}_{#1}(#2)}

In the reduction, we also use formula $\descr{\st_\QBFsymbol}{\setQBFvariableslevel\QBFsymbol}$ which expresses that variables in $X_s$ represent a given state~$S_0$:
\begin{align*}
    \descr{\st_0}{\setQBFvariableslevel\QBFsymbol} := & 
         \bigwedge_{i\in\AGT} \left(
            \bigwedge_{\alpha \in B_i} \propfortriangle{\expbel i \alpha}s
            \land
            \bigwedge_{\alpha \in \setrelevantformulas_i \setminus B_i} \lnot 
                \propfortriangle{\expbel i \alpha}\QBFsymbol
        \right) \\
        &
        \land
        \bigwedge_{p \in V} \propfortriangle p \QBFsymbol
        \land
        \bigwedge_{p \not \in V} \lnot \propfortriangle p \QBFsymbol. 
\end{align*}

The reduction from the model checking for  $\langdyn$ into TQBF is given in the following proposition.

\begin{proposition}\label{prop:traduction}
\label{prop:tqbftrans}
Let $\varphi_0 \in \langdyn$
and $\st_0$ be a state. We have $(\st_0,\setbelbasecustom) \models \varphi_0$  iff $\exists \setQBFvariableslevel\QBFsymbol (
            \descr{\st_0}{\setQBFvariableslevel\QBFsymbol} \land tr_\QBFsymbol(\varphi_0)
        )$ is QBF-true.
\end{proposition}
 
\subsection{Implementation}
\label{subsection-implementation}

\begin{table}
\centering
\begin{tabular}{lrrrrr}
    \hline
    $|\AGT|$ & 1 & 10 & 20 & 40 & 60\\
    $|\atmset|$ & $6$ & $60$ & $120$ & $240$ & $360$\\
    $|\Gamma_i|$ & $4$ & $4$ & $4$ & $4$ & $4$\\
    $|\setbelbasecustom|$ & $2^{18}$ & $2^{180}$ & $2^{360}$ & $2^{720}$ & $2^{1080}$\\
    \hline
    exec.\ time (sec.) & $0.07$ & $0.19$ & $0.57$ & $2.38$ & $5.67$\\
    \hline
\end{tabular}
\caption{\label{tab:exec-times}Model checker performance on Example~\ref{example:mother}.}
\end{table}

In order to verify its feasibility, we implemented a symbolic model checker for the static part of the language in Haskell (see Code URL on page~\pageref{code-url}) which uses the reduction to TQBF.
The resulting QBF-formula is solved using the binary decision diagram (BDD) library HasCacBDD \cite{HasCacBDD}.
It was compiled with GHC 9.4.8 in a MacBook Air with a 1.6 GHz Dual-Core Intel Core i5 processor and 16,537 GB of RAM, running macOS Sonoma 14.4.1.

Table~\ref{tab:exec-times} shows the performance of the implementation on different instances of Example \ref{example:mother}.
The execution times correspond to the elapsed time to model check Equation (\ref{check3}).
The number of states in the model $|\setbelbasecustom|$ gives an idea of the size of the task, and also the reason why a naive implementation would not be feasible.
The number of states corresponds to $2$ to the power of the number of variables in the set $X_s$ defined above.
The latter corresponds to the number of propositional variables plus $3$ times the number of formulas in each agent's belief base, i.e.,
$|\atmset| + 3 \times \big|\bigcup_{i \in \AGT}\Gamma_i\big|$.

\section{Conclusion}
\label{sec:conclusion}

We have presented a modal
logic framework
for 
cognitive attitudes
of agents
whose semantics relies  on 
belief
bases
and in which  models
are
represented succinctly. 
It 
allows us to reason about mental 
attitudes
of both epistemic
and motivational type
 including complete attraction and repulsion, realistic attraction
 and repulsion,
 and the derived concepts of motivation, demotivation,
 indifference, ambivalence and preference. 
We have supported our contribution 
with a number of results
about axiomatics 
and expressiveness, and with a 
theoretical
and experimental analysis of model checking
based on a reduction into TQBF. 

We  did not assume an agent's belief base to be  globally
consistent due to the presence of belief expansion operations in the dynamic setting. 
 Indeed, a belief expansion operation 
 could make an agent's belief base inconsistent.  For example, 
if an agent explicitly believes that $\alpha$
and that $\alpha \rightarrow \beta $,
after expanding its  belief base with $\neg \beta $,
its  belief base will become
inconsistent.
To  be able to handle belief dynamics while preserving the requirement that belief bases are globally  consistent, we 
should replace belief  expansion with  revision. However, this replacement  would make the overall semantics more complex and convoluted.  
Indeed,
as shown in \cite{DBLP:conf/ijcai/LoriniS21}, 
in order 
to model 
belief revision
in the belief base semantics 
we must compute  the intersection
of all  maximal consistent
subsets
of the belief base 
which contain the input formula. 
The interested reader can find 
the technical details of 
the extension  
with belief revision 
 as well as the reduction 
of its model checking problem 
into TQBF in
\OnlyAaai{the extended version of the paper (see Extended version URL on page~\pageref{extended-url}).}%
\OnlyArxiv{the appendix.}
Interestingly, the PSPACE model checking
procedure   extends
to revision in a natural way. 

The attraction and repulsion modalities
$\attract{i}{} $
and $\repuls {i}{} $
are symmetric with respect to the ``good''/``bad'' dimension.
A further interesting notion we plan
to consider in future work  is `strong' motivation.
Being `strongly' motivated by $\varphi$ is captured by the condition that every state satisfying $\varphi$ is both attractive and not repulsive.
In other words, for an agent to be strongly motivated by $\varphi$, it must be the case that the presence of $\varphi$ makes a situation attractive and prevents a situation from being repulsive.
Note that this notion of `strong' motivation is not definable from the other modalities.

The notion of preference
we have considered in the paper  is 
defined from purely qualitative notions of ``being motivated'' and ``being demotivated''. 
In  future work we also  plan to 
consider graded notions of being motivated/demotivated (i.e, the agent is motivated/demotivated by $ \varphi$ with a certain strength $ k\in \mathbb{N}$). This would allow us to refine the notion of preference, as follows: an agent prefers $ \varphi$ to $ \psi$ iff it is motivated by $ \varphi$ more than it is motivated by $ \psi$, or it is demotivated by $ \psi$ more than it is demotivated by $ \varphi$.

\section*{Acknowledgments}
Emiliano Lorini and François Schwarzentruber are supported by the ANR EpiRL project ANR-22-CE23-0029.
Elise Perrotin has  benefited from the support of the JSPS KAKENHI Grant Number JP21H04905.
Tiago de Lima is supported by the ANR AI Chair project ``Responsible AI'' (grant number ANR-19-CHIA-0008).

\bibliography{biblio}


\OnlyArxiv{\clearpage 

\appendix

\section{Appendix I: Proofs}

This  appendix contains the detailed proofs of
Theorems \ref{theo:unexpr}
and \ref{theo:2}
as well as the proofs of 
Proposition  \ref{prop:traduction}
about the translation into QBF,
and of 
all the other  propositions
given in the paper. 

\subsection{Proof of Theorem \ref{theo:unexpr}}\label{sec:ann1}

We give the proof of Theorem \ref{theo:unexpr} for the full language containing dynamic operators. 
We first show the result for $\attract{i}{} $ and $\repuls{i}{} $.

 \begin{proposition}\label{theo:unexpr1}
 The operators
     $\attract{i}{} $
     and
     $\repuls{i}{} $
     are not expressible
     with the other modalities.
 \end{proposition}

\begin{proof}
    We can w.l.o.g. place ourselves in the single agent case. Suppose that there exists a formula $\phi$ in which there are no $\attract i {}$ operators such that $\models\attract 1{} p\lequiv \phi$ for some $p\in\atmset$. Let $\atmset_\phi$ be the set of propositional variables occurring in $\phi$, and consider some $q,r\notin\atmset_\phi$. We now define $\iconstraint_1=\{S,S_1\}$ and $\iconstraint_2=\{S,S_2\}$ where:
    \begin{align*}
        S&=(\{q\limply\rewa 1,\lnot r\},\emptyset);\\
        S_1&=\{\{r\},\{p,q,r\}\};\\
        S_2&=\{\{r\},\{p,r\}\}.
    \end{align*}
    Observe that $S\relstate 1 S$, $S_1\relstate 1 S_1$ and $S_2\relstate 1 S_2$, but $S\cancel{\relstate 1} S_1,S_2$ and $S_1,S_2\cancel{\relstate 1} S$. Moreover, $S \relattr 1 S_1$ but $S \cancel{\relattr 1} S_2$.
    It follows that $S,\iconstraint_1\models\attract 1{}p$ but $S,\iconstraint_2\not\models\attract 1{}p$. However, $(S,\iconstraint_1)$ and $(S,\iconstraint_2)$ must agree on $\phi$.
    This is shown by proving by induction on the structure of $\psi$ that for any subformula $\psi$ of $\phi$, $(S,\iconstraint_1)\models\psi$ iff $(S,\iconstraint_2)\models\psi$, and $(S_1,\iconstraint_1)\models\psi$ iff $(S_2,\iconstraint_2)\models\psi$.
   
    The proof for $\repuls i{}$ is the same, replacing $\rewa 1$ with $\punis 1$.
\end{proof}

We now move on to $\attractreal{i}{}  $ and $ \repulsreal{i}{} $.
  \begin{proposition}
     The operators
     $\attractreal{i}{}  $
     and 
     $ \repulsreal{i}{} $
     are not expressible
    with the other modalities.
 \end{proposition}

\begin{proof}
    The proof is similar to that of Proposition \ref{theo:unexpr1}. We once again place ourselves in the single-agent case, and suppose that there exists a formula $\phi$ in which there are no $\attractreal i{}$ operators such that $\models\attractreal 1{} p\lequiv\phi$. We consider some $q,r,s\in\atmset\setminus\atmset_\phi$, and define $\iconstraint_1=\{S,S_1,S_2\}$ and $\iconstraint_2=\{S,S_1,S_2,S_3\}$ where:
    \begin{align*}
        S&=(\{q,r\limply\rewa 1\},\emptyset);\\
        S_1&=\{\{s\},\{p,q,r,\rewa 1\}\};\\
        S_2&=\{\{s\},\{p,\rewa 1\}\};\\
        S_3&=\{\{s\},\{p,q,\rewa 1\}\}.
    \end{align*}
    We have that $S\relstate 1 S_1,S_3$ and $S\relattr 1 S_1$; all alternative relations are otherwise empty.
    Moreover, $(S,\iconstraint_1)\models\attractreal 1{} p$ while $(S,\iconstraint_2)\not\models\attractreal 1{} p$. 
    However we can once again show by induction that $(S,\iconstraint_1)$ and $(S,\iconstraint_2)$ agree on all subformulas of $\phi$ and hence on $\phi$ itself.

    The proof for $\repulsreal i{}$ is the same, replacing the $\rewa 1$ with $\punis 1$.
\end{proof}

\subsection{Proof of Theorem \ref{theo:2} }
\label{sec:ann2}

We now give the proof of completeness of our axiomatization. We use similar techniques to the completeness proof in \cite{LoriniAI2020}.

In order to simplify the proof, and because attraction and repulsion are fully symmetrical, we place ourselves in the repulsion-free fragment, that is, we do not consider the $\repuls i{}$ and $\repulsreal i{}$ operators.
To prove completeness of the axiomatization we must define two classes of models called \textit{notional doxastic models with attraction} (NDMAs) and \textit{quasi-notional doxastic models with attraction} (quasi-NDMAs). A NDMA is a model $M=(W,B,E,A,RA,V)$ where $W$ is a set of worlds, $B:\agtset\times W\to 2^{\fraglang}$ indicates agents' explicit beliefs at each world, $E,A,RA:\agtset\times W\to 2^W$ are alternative relations between worlds for each agent, and $V:\atmset\to 2^W$ is a valuation, and such that under the semantics
\begin{eqnarray*}
    (M,w) \models p & \Longleftrightarrow & w\in V(p),\\
    (M,w) \models \expbel i\alpha & \Longleftrightarrow & \alpha\in B(i,w),\\
    (M,w)\models \impbel{i}\phi & \Longleftrightarrow
        & \forall v \in  W , \text{ if } v\in E(i,w)
            \text{ then }\\
&&             (M,v) \models \phi,\\
             (M,w) \models \attract{i}{} \varphi & \Longleftrightarrow & \forall v \in W , \text{if }
  (M,v) \models \varphi \text{ then } \\
  && v\in A(i,w),\\
   (M,w) \models \attractreal{i}{}\varphi 
 & \Longleftrightarrow & 
 \forall v \in  W  , \text{ if }
  (M,v) \models \varphi
  \text{ then }\\
  &&   v\in RA(i,w),
\end{eqnarray*}
and as usual for boolean operators, the following conditions hold:
\begin{align*}
    E(i,w)&=\bigcap_{\alpha\in B(i,w)}||\alpha||;
   \tag{NDMA$_1$}\\
    A(i,w)&=\bigcup_{\alpha\limply\rewa i\in B(i,w)}||\alpha||;
    \tag{NDMA$_2$}\\
    RA(i,w)&=A(i,w)\cup \overline{E}(i,w);
    \tag{NDMA$_3$}
\end{align*}
for all $i\in\agtset$ and $w\in W$, where $\overline{E}(i,w)=W\setminus E(i,w)$.

A quasi-NDMA is a similar model except that the three conditions are replaced by:
\begin{align*}
    E(i,w)&\subseteq\bigcap_{\alpha\in B(i,w)}||\alpha||;
   \tag{qNDMA$_1$}\\
    A(i,w)&\supseteq\bigcup_{\alpha\limply\rewa i\in B(i,w)}||\alpha||;
    \tag{qNDMA$_2$}\\
    RA(i,w)&\supseteq A(i,w)\cup \overline{E}(i,w);
    \tag{qNDMA$_3$}\\
    RA(i,w)&\subseteq \bigcup_{E(i,w)\subseteq||\phi\limply\rewa i||}||\phi||.
    \tag{qNDMA$_4$}
\end{align*}

Note that the four latter conditions are all implied by the combination of (NDMA$_1$), (NDMA$_2$) and (NDMA$_3$).
To show this for (qNDMA$_4$), take $v\in RA(i,w)$: then by (NDMA$_3$) either $v\in A(i,w)$ or $v\notin E(i,w)$. If the former, by (NDMA$_2$) there exists $\alpha\limply\rewa i\in B(i,w)$ such that $v\in||\alpha||$, and so we take $\phi=\alpha$ in (qNDMA$_4$). If the latter, then by (NDMA$_1$) there exists $\alpha\in B(i,w)$ such that $v\notin||\alpha||$, and we can take $\phi=\lnot\alpha$ in (qNDMA$_4$).

The completeness proof goes as follows. Given a finite set of formulas $\Sigma$ containing all $\rewa i$ and closed under the subformula relation as well as the implications 
$ \expbel{i }\alpha \rightarrow \impbel{i} \alpha $, $ \expbel{i }(\alpha\rightarrow \rewa{i} )\rightarrow \attract{i}{} \alpha$, $\attract{i}{}\varphi \rightarrow   \attractreal {i}{}\varphi$, $\impbel{i}{}\varphi \rightarrow   \attractreal i{} \neg \varphi$, and $\attractreal{i}{}\varphi \rightarrow  \impbel{i} (\varphi \rightarrow \rewa{i})$,
we define a finite quasi-NDMA which we call the $\Sigma$-canonical model such that any formula in $\Sigma$ is satisfiable in the logic $\mathsf{LCA}$ iff it is satisfiable in this model with the semantics described above. We then show how to define from a finite quasi-NDMA an equivalent finite quasi-NDMA satisfying property (NDMA$_3$), and from there how to define an equivalent NDMA, and from that an equivalent belief base context. By equivalent we mean that satisfiability of formulas in $\Sigma$ is preserved through the transformations. 

Given a set of formulas $\Sigma$ with the properties described above, the $\Sigma$-canonical model is defined as follows: $M_\Sigma=(W_\Sigma,B_\Sigma,E_\Sigma,A_\Sigma,RA_\Sigma,V_\Sigma)$ where
\begin{align*}
    W_\Sigma=\{X\cap&\Sigma\suchthat X\text{ is a maximal consistent set in }\mathsf{LCA}\};\\
    B_\Sigma(i,w)&=\{\alpha\suchthat\expbel i\alpha\in w\};\\
    E_\Sigma(i,w)&=\{v\suchthat\forall\impbel i\phi\in w, \phi\in v\};\\
    A_\Sigma(i,w)&=\{v\suchthat\exists\phi\in v,\attract i{}\phi\in w\};\\
    RA_\Sigma(i,w)&=\{v\suchthat\exists\phi\in v,\attractreal i{}\phi\in w\};\\
    V_\Sigma(p)&=\{w\suchthat p\in w\}.
\end{align*}

\begin{lemma}
For any formula $\phi\in\Sigma$ and any $w\in W_\Sigma$, $(M_\Sigma,w)\models \phi$ iff $\phi\in w$. Moreover, $M_\Sigma$ is a quasi-NDMA. 
\end{lemma}

\begin{proof}
    We show the first result by induction on the structure of $\phi$. The cases of atoms of $\fraglang$ and boolean operators are straightforward.

    If $\impbel i\phi\in w$ then it follows immediately from the definition of $E(i,w)$ that $(M,w)\models\impbel i\phi$. If $(M,w)\models\impbel i\phi$, where $w=X\cap\Sigma$ for some MCS $X$, let $A=\{\psi\suchthat\impbel i\psi\in w\}\cup\{\lnot\phi\}$. If $A$ is consistent then there exists a maximal consistent set $Y$ such that $A\subseteq Y$. Then by definition of $E(i,w)$ we have that $Y\cap\Sigma\in E(i,w)$, but by the induction hypothesis $(M,Y\cap\Sigma)\not\models \phi$. Hence $A$ is inconsistent, that is, $\vdash \bigwedge_{\psi\suchthat\impbel i\psi\in w}\psi\limply\phi$. As $X$ is an MCS we conclude that $\impbel i\phi\in X$, hence $\impbel i\phi\in w$.

    If $\attract i{}\phi\in w$, suppose that $M,v\models\phi$ for some $v\in W$. Then by the induction hypothesis $\phi\in v$, hence by the definition of $A_\Sigma(i,w)$ we have $v\in A_\Sigma(i,w)$. Therefore $M,w\models \attract i{}\phi$.

    If $(M,w)\models\attract i{}\phi$, where $w=X\cap\Sigma$ for some MCS $X$, consider the formula $\Phi=\phi\land\lnot\bigvee_{\attract i{}\psi\in w}\psi$. If it is consistent then there exists a MCS $Y$ such that $\Phi\in Y$, hence $\phi\in Y\cap\Sigma$ and $\psi\notin Y\cap\Sigma$ for all $\psi$ such that $\attract i{}\psi\in w$. By the induction hypothesis, we get that $(M,Y\cap\Sigma)\models \phi$, hence as $(M,w)\models\attract i{}\phi$ we obtain that $Y\cap\Sigma\in A(i,w)$. This contradicts the definition of $A(i,w)$. Therefore $\Phi$ is inconsistent, that is, $\vdash\lnot(\phi\land\lnot\bigvee_{\attract i{}\psi\in w}\psi)$. From the rule of necessitation for $\attract i{}$ and uniform substitution we get that $\vdash\attract i {}(\phi\land\lnot\bigvee_{\attract i{}\psi\in w}\psi)$. Moreover it can be shown that $\vdash \attract i {}\psi \land \attract i {}\psi'\limply \attract i{}(\psi\lor\psi')$.%
    \footnote{Here is a derivation:
    \begin{align*}
        &\vdash \lnot (\lnot \psi \land (\lnot\phi\land\psi))\\
        &\vdash \attract i{}(\lnot \psi \land (\lnot\phi\land\psi)) 
        \\
        &\vdash \attract i{}\psi \land \attract i{}(\lnot \psi \land (\lnot\phi\land\psi)) \limply \attract i{} (\lnot\phi\land\psi) 
        \\
        &\vdash \attract i{}\psi\limply \attract i{} (\lnot\phi\land(\phi\lor\psi)) 
        \\
        &\vdash \attract i{}\phi\land \attract i{} (\lnot\phi\land(\phi\lor\psi))\limply \attract i{} (\phi\lor\psi)
        \\
        &\vdash \attract i{}\phi\land\attract i{}\psi\limply \attract i{} (\phi\lor\psi) 
    \end{align*}
    }
    
    Generalizing this we obtain that $\attract i {}(\bigvee_{\attract i{}\psi\in w}\psi)\in X$. From the axiom $\big( \attract i{}    \varphi \wedge \attract i{} (\neg \varphi \wedge \psi)  \big) \rightarrow \attract i{}  \psi$ we finally obtain that $\attract i{}\phi\in X$, hence $\attract i{}\phi\in w$.

    The case of $\attractreal i{}\phi$ is proved in the same way as that of $\attract i{}\phi$.
    
    We now show that $M_\Sigma$ is a quasi-NDMA. Property (qNDMA$_1$) follows from the axiom $ \expbel{i }\alpha \rightarrow \impbel{i} \alpha $. Property (qNDMA$_2$) follows from the axiom $ \expbel{i }(\alpha\rightarrow \rewa{i} )\rightarrow \attract{i}{} \alpha$. Property (qNDMA$_3$) follows from the axioms $\attract{i}{}\varphi \rightarrow   \attractreal {i}{}\varphi$ and $\impbel{i}{}\varphi \rightarrow   \attractreal i{} \neg \varphi$. Property (qNDMA$_4$) follows from the axiom $\attractreal{i}{}\varphi \rightarrow  \impbel{i} (\varphi \rightarrow \rewa{i})$.
\end{proof}

Given a finite quasi-NDMA $M=(W,B,E,A,RA,V)$, we now define the following expansion of $M$: $M'=(W,B',E',A',RA',V')$ where

\begin{align*}
    W'=\ &\{w_1,w_2\suchthat w\in W\};\\
    B'(i,w_k)=\ &B(i,w);\\
    E'(i,w_k)=\ &\{v_1\suchthat v\in E(i,w)\}\cup
    \\&
    \{v_2\suchthat v\in E(i,w)\setminus (RA(i,w)\setminus A(i,w))\};\\
    A'(i,w_k)=\ &\{v_1,v_2\suchthat v\in A(i,w)\}\cup
    \\& \{v_1\suchthat v\in RA(i,w)\cap(E(i,w)\setminus A(i,w))\};\\
    RA'(i,w_k)=\ &\{v_1,v_2\suchthat v\in RA(i,w)\};\\
    V'(p)=\ &\{w_1,w_2\suchthat w\in V(p)\}
\end{align*}
for $k\in\{1,2\}$. That is, we create two copies of each world in $W$ and maintain the properties of all worlds except that when 
$v\in RA(i,w)\cap(E(i,w)\setminus A(i,w))$ we now have $v^1\in E'(i,w_k)\cap A'(i,w_k)$ and $v^2\notin E'(i,w_k)\cup A'(i,w_k)$ for $k\in\{1,2\}$.

\begin{lemma}
    For any formula $\phi$ and any $w\in W$, $(M,w)\models \phi$ iff $(M',w_1)\models\phi$ iff $(M',w_2)\models\phi$. Moreover $M'$ is a quasi-NDMA satisfying (NDMA$_3$).
\end{lemma}

\begin{proof}
    We show the first result by induction on the structure of $\phi$. The cases of $\fraglang$ atoms and boolean operators as well as $\attractreal i{}\phi$ are straightforward.

    If $(M,w)\models \impbel i\phi$ and $v_l\in E'(i,w_k)$ for some $k,l\in\{1,2\}$ then by definition of $E'$ we have $v\in E(i,w)$, hence $(M,v)\models \phi$ and by the induction hypothesis $(M',v_l)\models\phi$. Therefore $(M',w_k)\models\impbel i\phi$ for $k\in\{1,2\}$.

    If $(M,w_k)\models\impbel i\phi$ for some $k\in\{1,2\}$ and $v\in E(i,w)$ then $v_1\in E'(i,w_k)$, hence $(M',v_1)\models\phi$ and by the induction hypothesis $(M,v)\models\phi$. Therefore $(M,w)\models\impbel i\phi$.

    If $(M,w)\models \attract i{}\phi$ and $(M',v_l)\models\phi$ for some $l\in\{1,2\}$ then by the induction hypothesis $(M,v)\models\phi$, hence $v\in A(i,w)$. Then by the definition of $A'$ we have that $v_l\in A'(i,w_k)$ for all $k\in\{1,2\}$. Hence $(M',w_k)\models \attract i{}\phi$ for $k\in\{1,2\}$.

    If $(M,w_k)\models\attract i{}\phi$ for some $k\in\{1,2\}$ and $(M,v)\models\phi$ then by the induction hypothesis $(M,v_2)\models\phi$ and so $v_2\in A'(i,w_k)$. By definition of $A'$ this means that $v\in A(i,w)$. Therefore $(M,w)\models\attract i{}\phi$.

    Verifying the quasi-NDMA properties and (NDMA$_3$) is straightforward.
\end{proof}

We now wish to recover the properties (NDMA$_1$) and (NDMA$_2$). For this, given a finite quasi-NDMA $M=(W,B,E,A,RA,V)$ satisfying (NDMA$_3$) and such that $B(i,w)\subseteq\Sigma$ for all $i$ and $w$, we define the model $M'=(W,B',E,A,RA,V')$ as follows:
\begin{align*}
    B'(i,w)&=B(i,w)\cup \{p_{i,w},q_{i,w}\limply\rewa i\};\\
    V'(p)&=\begin{cases}
        V(p) &\text{ if }p\in\Sigma;\\
        E(i,w) &\text{ if }p=p_{i,w};\\
        A(i,w) &\text{ if }p=q_{i,w};\\
        \emptyset &\text{ otherwise;}
    \end{cases}
\end{align*}
where all $p_{i,w}$ and $q_{i,w}$ are fresh.

\begin{lemma}
    For any $\phi\in\Sigma$ and $w\in W$, $M,w\models\phi$ iff $M',w\models\phi$. Moreover $M'$ is a NDMA.
\end{lemma}

\begin{proof}
    The first result is shown by a simple induction on the structure of $\phi$; intuitively, as the alternative relations are unchanged and the newly introduced variables are fresh, satisfaction of formulas in $\Sigma$ is unaffected. It remains to show that $M'$ is a NDMA. Note that, once again as the alternative relations are unchanged, property (NDMA$_3$) is trivially preserved from $M$ to $M'$.

    Let $v\in E(i,w)$ and $\alpha\in B'(i,w)$. If $\alpha\in B(i,w)$ then $M,v\models\alpha$ by (qNDMA$_1$), hence $M',v\models\alpha$. If $\alpha=p_{i,w}$ then $M',v\models\alpha$ by definition of $V'$. Finally, if $\alpha=q_{i,w}\limply\rewa i$, there are two cases. If $M',v\not\models q_{i,w}$ then $M',v\models\alpha$. Otherwise by definition of $V'$ we have that $v\in A(i,w)$, hence by (qNDMA$_3$) $v\in RA(i,w)$. By (qNDMA$_4$) this implies that there exists a formula $\phi$ such that $M,v\models\phi$ and $M,u\models\phi\limply\rewa i$ for all $u\in E(i,w)$. As $v\in E(i,w)$, we obtain that $M,v\models\rewa i$. As $\rewa i\in\Sigma$ we conclude that $M',v\models\rewa i$ and $M',v\models\alpha$.

    Conversely, let $v\in W$ be such that $M',v\models\alpha$ for all $\alpha\in B'(i,w)$. Then in particular $M',v\models p_{i,w}$. By definition of $V'$ this implies that $v\in E(i,w)$.

    We move on to property (NDMA$_2$). If $v\in A(i,w)$ then by definition of $V'$ we have $M',v\models q_{i,w}$, and $q_{i,w}\limply\rewa i\in B'(i,w)$. Conversely,
    if there exists $\alpha\limply\rewa i\in B'(i,w)$ such that $M',v\models\alpha$ then there are two cases. If $\alpha\limply\rewa i\in B'(i,w)$ then $\alpha\in\Sigma$ and $M,v\models\alpha$, hence by (qNDMA$_2$) $v\in A(i,w)$. Otherwise $\alpha=q_{i,w}$ and by definition of $V'$ we also get that $v\in A(i,w)$.
\end{proof}

Finally, from a NDMA $M=(W,B,E,A,RA,V)$ we define the context $\iconstraint=\{\st^w\suchthat w\in W\}$, where $ \st^w = \big( (\belbaseset_i^w)_{i \in \AGT}, \stateval^w \big)$ is defined as follows:
\begin{align*}
    \belbaseset_i^w&=B(i,w);\\
    \stateval^w&=\{p\suchthat w\in V(p)\}.
\end{align*}

\begin{lemma}
For any formula $\phi$ and any $w\in W$, $(M,w)\models \phi$ iff $(\st^w,\iconstraint)\models\phi$.

\begin{proof}
    This is shown by induction on the structure of $\phi$. The cases of $\fraglang$ atoms and boolean operators are straightforward.

    Following this and using (NDMA$_1$), it is easy to prove that for any worlds $w$ and $v$, $\st^w\relstate i \st^w$ iff $v\in E(i,w)$, from which follows the case of $\impbel i\phi$.
    Similarly, using (NDMA$_2$) we get that $\st^w\relattr i \st^w$ iff $v\in A(i,w)$, from which follows the case of $\attract i{}\phi$.

    Suppose now that $(M,w)\models\attractreal i{}\phi$, and let $v\in W$ be such that $\st^v\models\phi$ and $\st^w\relstate i \st^v$. Then we have seen that $v\in E(i,w)$, and by the induction hypothesis $(M,v)\models\phi$. As $(M,w)\models\attractreal i{}\phi$ we get that $v\in RA(i,w)$, hence by (NDMA$_3$) $v\in A(i,w)$. Therefore $\st^w\relattr i \st^w$, and we conclude that $(\st^w,\iconstraint)\models\attractreal i{}\phi$.

    Conversely, if $(\st^w,\iconstraint)\models\attractreal i{}\phi$ and consider $v\in W$ such that $(M,v)\models\phi$. If $v\notin E(i,w)$ then $v\in RA(i,w)$ by (NDMA$_3$). Otherwise we have seen that $\st^w\relstate i \st^v$, and by the induction hypothesis $\st^v\models\phi$, hence $\st^w\relattr i \st^v$ and $v\in A(i,w)$. Therefore by (NDMA$_3$) $v\in RA(i,w)$. We conclude that $(M,w)\models\attractreal i{}\phi$.
\end{proof}
\end{lemma}

\begin{proposition}
    If a formula $\phi$ is satisfiable in $\mathsf{LCA}$ then it is satisfiable in $\classbelbase$.
\end{proposition}

\begin{proof}
    Let $\phi$ be a formula satisfiable in $\mathsf{LCA}$. We take as $\Sigma$ the set of subformulas of $\phi$ augmented with all $\rewa i$ and closed under the relevant implications. Then $\phi$ is satisfied in some world $w$ of the $\Sigma$-canonical model, hence it is also satisfied in the two expansions we have defined, and finally it is satisfied in the corresponding model $(\st^w,\iconstraint)$.
\end{proof}

\subsection{Proof  of Proposition \ref{prop:traduction}}\label{sec:proofPropTraduction}

\begin{proof}
\newcommand{\valuationdescr}[2]{\operatorname{val}_{#1}(#2)}
Let $\valuationdescr{\st}{\setQBFvariableslevel k}$
represent the unique valuation on $X_k$ satisfying $\descr{\st}{\setQBFvariableslevel k}$.
We prove by induction on the structure of $\phi$ that
$%
    (\st,\setbelbasecustom) \models \phi
    \text{ iff }
    \valuationdescr{\st}{\setQBFvariableslevel k} \models tr_\QBFsymbol(\phi)
$,
for all $k$.

Induction base.
Let $\phi = p$, for some $p \in \atmset$.
We have
$(\st,\setbelbasecustom) \models p$
iff
$p \in V$
iff    
$x_{k, p} \in \valuationdescr{\st}{\setQBFvariableslevel k}$
iff
$\valuationdescr{\st}{\setQBFvariableslevel k} \models tr_\QBFsymbol(p)$.

Induction step.
The cases for operators $\lnot$ and $\land$ are straightforward.
We proceed with the modal operators in the language:
\begin{itemize}
\item Let $\phi = \expbel i \alpha$.
    We have:
    $(\st,\setbelbasecustom) \models \expbel i \alpha$
    iff
    $\alpha \in B_i$
    iff
    $x_{k, \expbel i \alpha} \in \valuationdescr{\st}{\setQBFvariableslevel k}$
    iff
    $\valuationdescr{\st}{\setQBFvariableslevel k} \models tr_\QBFsymbol(\expbel i \alpha)$.
    
\item
    Let $\phi = \impbel i \psi$.
    We denote by $\valuationdescr{\st}{\setQBFvariableslevel {\QBFsymbol}} + \valuationdescr{\st'}{\setQBFvariableslevel {\QBFfreshsymbol}}$ the valuation obtained by concatenating the valuation $\valuationdescr{\st}{\setQBFvariableslevel {\QBFsymbol}}$ and $\valuationdescr{\st'}{\setQBFvariableslevel {\QBFfreshsymbol}}$ (we take the truth values of propositions in $\setQBFvariableslevel {k}$ from the former and the truth values of propositions in $\setQBFvariableslevel {\QBFfreshsymbol}$ from the latter).
    We have:
\begin{align*}
    (\st,\setbelbasecustom) \models  \impbel i \psi
    \iff & 
    \text{ for all $\st' \in \setbelbasecustom$, $\st \relepist{i} \st'$ implies } \\
    & (\st',\setbelbasecustom) \models  \psi
    \\
    \iff & 
    \text{ for all $\st' \in \setbelbasecustom$, $\st \relepist{i} \st'$ implies } \\ 
    & \valuationdescr{\st'}{\setQBFvariableslevel {\QBFfreshsymbol}} \models tr_{\QBFfreshsymbol}(\psi)
    \\
    \iff & 
    \text{ for all $\st' \in \setbelbasecustom$} \\
    & 
    \valuationdescr{\st}{\setQBFvariableslevel {\QBFsymbol}} + \valuationdescr{\st'} {\setQBFvariableslevel {\QBFfreshsymbol}} \\ 
    & ~~~~~~\models R_{i,\QBFsymbol} \limply  tr_{k+1}(\psi) \\
    \iff & 
    \valuationdescr{\st}{\setQBFvariableslevel {\QBFsymbol}} \models \forall \setQBFvariableslevel {\QBFfreshsymbol}, R_{i,\QBFsymbol} \limply  tr_{\QBFfreshsymbol}(\psi) \\
    \iff & 
    \valuationdescr{\st}{\setQBFvariableslevel \QBFsymbol} \models tr_\QBFsymbol(\impbel i \alpha)
\end{align*}
   \item The other cases are analogous.
\end{itemize}

Therefore, $(\st_0,\setbelbasecustom) \models \phi_0$ iff  $\valuationdescr{\st_0}{\setQBFvariableslevel 0} \models tr_0(\phi)$.
In addition,
$\valuationdescr{\st_0}{\setQBFvariableslevel 0} \models tr_0(\phi)$ is equivalent to $\exists \setQBFvariableslevel0 (
            \descr{\st_0}{\setQBFvariableslevel0} \land tr_0(\phi_0)
        )$ is QBF-true.
This concludes the proof.
\end{proof}

\subsection{Proofs of the other  propositions}\label{sec:proofProp}

We prove validity  (\ref{motdemot1}) in  Proposition \ref{prop:motdemot}.
The proof
of validity  (\ref{motdemot2}) is analogous. 

\subsubsection{Proof of validity  (\ref{motdemot1}) in  Proposition \ref{prop:motdemot} }
\begin{proof}
Suppose $ (\st,\iconstraint) \models\motiv{i}\phi   $
for an arbitrary model $ (\st,\iconstraint)$.
Thus, by definition of $\motiv{i}$, 
$ (\st,\iconstraint)  \models  \neg \repuls{i}{} \phi$.
Hence, by definition of $\demotiv{i}$, 
$ (\st,\iconstraint)  \models  \neg \demotiv{i}\phi$.
\end{proof}

Then, we 
prove the three validities 
in  Proposition  \ref{prefprop}.
We only prove the case
$\jokertwo =\prec_i$.
The proof of the  case 
$\jokertwo = \prec_i^{\mathsf{real}}$
is 
analogous.

\subsubsection{Proof of validity  (\ref{propirr}) in  Proposition \ref{prefprop}}
\begin{proof}
Clearly, $\neg (\motiv{i}\varphi  \wedge  \neg \motiv{i}\varphi) $
is valid. Thus, by definition of 
$\prec_i$, 
$\neg  (\phi   \prec_i  \phi )$
is valid too.
\end{proof}

\subsubsection{Proof of validity  (\ref{propasy}) in Proposition \ref{prefprop}}
\begin{proof}
Suppose $ (\st,\iconstraint) \models\psi   \prec_i  \phi  $
for an arbitrary model $ (\st,\iconstraint)$.
By definition of $\prec_i$,
two cases are possible.

\textbf{Case $ (\st,\iconstraint) \models
\motiv{i}\varphi  \wedge  \neg \motiv{i}\psi $.}
By item (\ref{motdemot1}) of Proposition \ref{prop:motdemot},
it follows that 
$ (\st,\iconstraint) \models
\motiv{i}\varphi  \wedge  \neg \motiv{i}\psi
\wedge \neg 
\demotiv{i}\varphi  $. Thus,  by definition of $\prec_i$,
$(\st,\iconstraint) \models \neg  (\phi   \prec_i  \psi)$.

\textbf{Case $ (\st,\iconstraint) \models
\demotiv{i}\psi   \wedge  \neg \demotiv{i}\varphi  $.}
By item (\ref{motdemot1}) of Proposition \ref{prop:motdemot},
it follows that 
$ (\st,\iconstraint) \models
\demotiv{i}\psi   \wedge  \neg \demotiv{i}\varphi
\wedge \neg 
\motiv{i}\psi   $. Thus, again by definition of $\prec_i$,
$(\st,\iconstraint) \models \neg  (\phi   \prec_i  \psi)$.
\end{proof}

\subsubsection{Proof of validity  (\ref{proptrans}) in  Proposition \ref{prefprop}}
\begin{proof}
Suppose $ (\st,\iconstraint) \models(\phi_1   \prec_i  \phi_2 )
\wedge (\phi_2   \prec_i  \phi_3 )$
for an arbitrary model $ (\st,\iconstraint)$.
By definition of $\prec_i$,
four cases should be considered. 

\textbf{Case $ (\st,\iconstraint) \models
\neg \motiv{i}\phi_1  \wedge  \motiv{i}\phi_2
\wedge \neg \motiv{i}\phi_2  \wedge   \motiv{i}\phi_3 $.}
This case is not possible since 
$ \motiv{i}\phi_2
\wedge \neg \motiv{i}\phi_2$ 
is unsatisfiable.

\textbf{Case $ (\st,\iconstraint) \models
 \demotiv{i}\phi_1  \wedge \neg  \demotiv{i}\phi_2
\wedge  \demotiv{i}\phi_2  \wedge \neg   \demotiv{i}\phi_3 $.}
This case is not possible since 
$  \neg  \demotiv{i}\phi_2
\wedge  \demotiv{i}\phi_2$ 
is unsatisfiable.

\textbf{Case $ (\st,\iconstraint) \models
\neg \motiv{i}\phi_1  \wedge  \motiv{i}\phi_2
\wedge  \demotiv{i}\phi_2  \wedge  \neg \demotiv{i}\phi_3 $.}
This case is not possible since
$ \motiv{i}\phi_2
\wedge  \demotiv{i}\phi_2 $ 
is unsatisfiable
due to item (\ref{motdemot1}) of Proposition \ref{prop:motdemot}.

\textbf{Case $ (\st,\iconstraint) \models
 \demotiv{i}\phi_1  \wedge \neg  \demotiv{i}\phi_2
\wedge  \neg \motiv{i}\phi_2  \wedge  \motiv{i}\phi_3 $.}
  By item (\ref{motdemot1}) of Proposition \ref{prop:motdemot},
  it implies 
  $ (\st,\iconstraint) \models
\neg  \motiv{i}\phi_1  \wedge  \motiv{i}\phi_3 $.
Thus, by definition of $ \prec_i$, 
$ (\st,\iconstraint) \models
\phi_1  \prec_i  \phi_3 $. 
\end{proof}

\newpage 

\section{Appendix II: Dynamic extension with revision}\label{sec:revision}

\newcommand{\beliefbase}[1]{B_{#1}}

This  appendix
contains the technical
details
of the dynamic extension
of the language
$\langdyn$
with belief
revision operators
of type  $*_i \alpha$, 
as well as  the reduction
of its model checking problem into TQBF, 
we mentioned in the conclusion of the paper. 

\subsection{Semantics}
The semantics of the revision operator $[*_i \alpha]$ is given by:
\begin{align*}
(\st,\iconstraint)\models [*_i \alpha] \phi & \text{ iff }
(\st',\iconstraint)\models  \phi
\end{align*}
where $\st'$ is a state where the valuation is the same as in $\st$, i.e., $V' = V$;
the belief bases of all agents different from $i$ are the same as in $\st$, i.e., $\beliefbase j' = \beliefbase j$, for $j \neq i$;
and $\beliefbase i'$ is the intersection of the maximal consistent sets (MCS) included in $\beliefbase i \union \set{\alpha}$, which we call $\alpha$-MCS.

\subsection{Abbreviations}

The QBF translation approach can still be applied. To this aim, we introduce several abbreviations.
In what follows, let $B_i(s)$ be the $i$'s belief base in the state represented by the symbol $s$.

\newcommand{\imcs}{\operatorname{IMCS}}

\begin{itemize}

\newcommand{\included}{incl}
\item The formula $\included_i(X_s, X_{s'})$ defined below says that the set $B_i(s)$ is included in $B_i(s')$:
\[
    \included_i(X_s, X_{s'}) := \lbigand_{\alpha \in \Gamma_i} x_{\expbel i \alpha, s} \limply x_{\expbel i \alpha, s'}
\]

\newcommand{\strictlyincluded}{strincl}
\item The formula $\strictlyincluded_i(X_s, X_{s'})$ defined below says that $B_i(s)$ is strictly included in $B_i(s')$:
\begin{align*}
    \strictlyincluded_i(X_x, X_{s'}) := \;
    & \included_i(X_s, X_{s'}) \;\land \\
    & \lnot \lbigand_{\alpha \in \Gamma_i} (x_{\expbel i \alpha, s} \lequiv x_{\expbel i \alpha, s'})
\end{align*}

\newcommand{\consist}{cons}
\item The formula $\consist_i(X_s)$ defined below says that $B_i(s)$ is consistent:
\[
    \consist_i(X_s) := tr_i(\lnot \impbel i \bot)
\]

\newcommand{\includedminusalpha}{inclminus}
\item The formula $\includedminusalpha_i(\alpha, X_{s}, X_{s'})$ says that the set
$B_i(s) \setminus \set{\alpha}$ is included in $B_i(s')$
(which is equivalent to say that $B_i(s)$ is included in $B_i(s') \union \set{\alpha}$):
\[
    \includedminusalpha_i(\alpha,X_s,X_{s'}) 
    :=
    \lbigand_{\beta \in \Gamma_i, \beta \neq \alpha} (x_{\expbel i \beta, {s}} \limply x_{\expbel i \beta, {s'}})
\]

\newcommand{\mcs}{\operatorname{MCS}}

\item Now, consider fresh symbols $M$ and $M'$ representing states (as well as $s$, $s'$, etc., but with the difference that $M$ and $M'$ are not symbols that denote states that are present in the model).
The formula $\mcs_i(X_M, X_s, \alpha)$ below says that $B_i(M)$ is a maximal consistent set included in $B_i(s) \union \set{\alpha}$:
\begin{align*}
    \mcs_i(X_M,X_s,\alpha) & := x_{\expbel i \alpha, M} \;\land\\
    & \includedminusalpha_i(\alpha, X_M, X_s) \;\land\\
    & \consist_i(X_M) \;\land\\
    & \forall X_{M'} \begin{aligned}[t]
        ( \; & \strictlyincluded_i(X_M, X_{M'}) \;\land \\
        & \includedminusalpha_i(\alpha,X_{M'}, X_s)\\
        ) & \limply \lnot \consist_i(X_{M'}) 
    \end{aligned}
\end{align*}

The first line of the formula above says that $B_i(M)$ contains $\alpha$;
the second line says that $B_i(M)$ is included in $B_i(s) \union \set{\alpha}$;
the third line says that $B_i(M)$ is consistent;
and
the remainder of the formula says that any proper superset $B_i(M')$ of $B_i(M)$ included in $B_i(s) \union \set{\alpha}$ is inconsistent.

\item Finally, the formula $\imcs_i(X_s, X_{s'}, \alpha)$ below says that $B_i(s')$ equals to the intersection of all $\alpha$-MCS:
\begin{align*}
    & \imcs_i(X_s, X_{s'}, \alpha) := \\
    & \quad ( \forall X_M (\mcs_i(X_M, X_s, \alpha) \limply \included_i(X_{s'}, X_M) ) \;\land \\
    & \quad \lbigand_{\beta \in \Gamma_i} (\forall X_M (\mcs_i(X_M, X_s, \alpha) \limply x_{\expbel i \beta, M})) \limply x_{\expbel i \beta, s'}
\end{align*} 

The first line of the formula above says that any $\alpha$-MCS includes $B_i(s')$.
The second line is the other direction, i.e., if $\beta$ is in every $\alpha$-MCS then $\beta$ is in $B_i(s')$.
\end{itemize}

\subsection{Modification of the translation function}

Now we can add the following clause to the definition of $\trrel \QBFsymbol \QBFsymbolb$:
\begin{align*}
    \trrel  \QBFsymbol \QBFsymbolb (*_i \alpha) := \;
    & \sameprop\QBFsymbol\QBFsymbolb \;\land \\
    & \lbigand_{j \neq i} \samej\QBFsymbol\QBFsymbolb \;\land \\
    & \imcs_i(X_s, X_{s'}, \alpha) 
\end{align*}
The first line of the formula above says that the valuations in $s$ and $s'$ are the same;
the second line says that, except for agent $i$, all agent's belief bases are the same in $s$ and $s'$;
and
the third line says that $B_i(s')$ is the intersection of all $\alpha$-MCS.
}

\end{document}